\documentclass[10pt]{article}

\usepackage{amstext,amssymb,amsthm,amsmath} \usepackage{natbib,booktabs,,graphicx,rotating,array}
\usepackage{algorithmic,algorithm}
\usepackage[margin=1in]{geometry}

\newtheorem{theorem}{Theorem}[section]
\newtheorem{lemma}[theorem]{Lemma}

\newtheorem{proposition}[theorem]{Proposition}

\newcommand{\bd}{\boldsymbol{d}}
\newcommand{\be}{\boldsymbol{e}}

\newcommand{\br}{\boldsymbol{r}}
\newcommand{\bs}{\boldsymbol{s}}
\newcommand{\bt}{\boldsymbol{t}}
\newcommand{\bu}{\boldsymbol{u}}
\newcommand{\bv}{\boldsymbol{v}}
\newcommand{\bw}{\boldsymbol{w}}
\newcommand{\bx}{\boldsymbol{x}}
\newcommand{\by}{\boldsymbol{y}}

\newcommand{\bA}{\boldsymbol{A}}
\newcommand{\bB}{\boldsymbol{B}}
\newcommand{\bC}{\boldsymbol{C}}

\newcommand{\bH}{\boldsymbol{H}}
\newcommand{\bI}{\boldsymbol{I}}
\newcommand{\bK}{\boldsymbol{K}}
\newcommand{\bP}{\boldsymbol{P}}
\newcommand{\bQ}{\boldsymbol{Q}}
\newcommand{\bM}{\boldsymbol{M}}
\newcommand{\bR}{\boldsymbol{R}}
\newcommand{\bX}{\boldsymbol{X}}
\newcommand{\bU}{\boldsymbol{U}}
\newcommand{\bV}{\boldsymbol{V}}
\newcommand{\bW}{\boldsymbol{W}}
\newcommand{\bY}{\boldsymbol{Y}}

\newcommand{\bbeta}{\boldsymbol{\beta}}

\newcommand{\blambda}{\boldsymbol{\lambda}}

\newcommand{\bphi}{\boldsymbol{\phi}}
\newcommand{\bmu}{\boldsymbol{\mu}}
\newcommand{\bOmega}{\boldsymbol{\Omega}}
\newcommand{\bSigma}{\boldsymbol{\Sigma}}

\title{\bf A Generic Path Algorithm for Regularized Statistical Estimation}
\author{
Hua Zhou \\
Department of Statistics \\
North Carolina State University \\
Raleigh, NC 27695-8203 \\
E-mail: hua\_zhou@ncsu.edu \\
\and
Yichao Wu \\
Department of Statistics \\
North Carolina State University \\
Raleigh, NC 27695-8203 \\
E-mail: wu@stat.ncsu.edu
}

\begin{document}
\maketitle

\baselineskip=20pt

\begin{abstract}
Regularization is widely used in statistics and machine learning to prevent overfitting and gear solution towards prior information. In general, a regularized estimation problem minimizes the sum of a loss function and a penalty term. The penalty term is usually weighted by a tuning parameter and encourages certain constraints on the parameters to be estimated. Particular choices of constraints lead to the popular lasso, fused-lasso, and other generalized $l_1$ penalized regression methods. Although there has been a lot of research in this area, developing efficient optimization methods for many nonseparable penalties remains a challenge. In this article we propose an exact path solver based on ordinary differential equations (EPSODE) that works for any convex loss function and can deal with generalized $l_1$ penalties as well as more complicated regularization such as inequality constraints encountered in shape-restricted regressions and nonparametric density estimation. In the path following process, the solution path hits, exits, and slides along the various constraints and vividly illustrates the tradeoffs between goodness of fit and model parsimony. In practice, the EPSODE can be coupled with AIC, BIC, $C_p$ or cross-validation to select an optimal tuning parameter. Our applications to generalized $l_1$ regularized generalized linear models, shape-restricted regressions, Gaussian graphical models, and nonparametric density estimation showcase the potential of the EPSODE algorithm.
\vspace{.1in}   
\vspace{.1in}
\end{abstract}

\section{Introduction}
\label{sec:intro}

Regularization is a frequently used framework in statistics. Examples include the lasso regression \citep{Tibshirani96Lasso,ChenDonohoSaunders01BasisPursuit} and the $l_1$ penalized generalized linear models (GLMs) among many others. For both the lasso and the $l_1$ penalized GLMs, efficient solution path algorithms have been proposed to ease the tuning of the regularization parameter \citep{OsbornePresnellTurlach00LassoAlgo,EfronHastieIainTibshirani04LARS,ParkHastie07GLMLasso}. Yet extension to other more general settings is nontrivial and has been an active research area.

In this article, we consider a general regularization framework
\begin{align}
    \min_{\bbeta \in \mathbb{R}^p} f(\bbeta) + \rho \|\bV \bbeta - \bd\|_1 + \rho \|\bW \bbeta - \be\|_+, \label{eqn:pathalgo-obj}
\end{align}
for which we propose an efficient exact path solver based on ordinary differential equations (EPSODE).
Here $f: \mathbb{R}^{p} \mapsto \mathbb{R}$ can be any convex, smooth loss function of $\bbeta\in\mathbb{R}^p$, where $p>0$ is the dimensionality of the parameters. For any vector $\bv=(v_i)$,  $\|\bv\|_1 = \sum_i |v_i|$ denotes its $l_1$ norm and $\|\bv\|_+ = \sum_i \max \{v_i,0\}$ is the sum of positive parts of its components. The EPSODE provides the exact solution path to (\ref{eqn:pathalgo-obj}) as the tuning parameter $\rho$ varies.

\subsection{Generality of (\ref{eqn:pathalgo-obj})}

The generality of (\ref{eqn:pathalgo-obj}) is two-fold. First  $f$ can by any convex loss function. For example, it can be the negative log-likelihood function of GLMs, negative quasi-likelihood, the exponential loss function of the AdaBoost \citep{friedman00additive}, or many other frequently used loss functions in statistics and machine learning. Second we allow $\bV$ and $\bW$ to be any regularization matrices of $p$ columns. This leads to broad applications. In particular, the first regularization term $\rho \|\bV \bbeta - \bd\|_1$ encourages equality constraints among parameters $\bbeta$. When $\rho$ is large enough, the minimizer $\bbeta(\rho)$ of (\ref{eqn:pathalgo-obj}) satisfies $\bV \bbeta(\rho) = \bd$. For instance, when $\bV$ is the identity matrix and $\bd={\bf 0}$, it recovers the well-known lasso regression \citep{Tibshirani96Lasso,ChenDonohoSaunders01BasisPursuit}, which encourages sparsity of the estimates. When
\begin{align*}
    \bV &= \left( \begin{array}{rrrrr}
    -1 & 1 &  \\
     & \ddots & \ddots  \\
     & & -1 & 1
    \end{array}  \right)
\end{align*}
and $\bd={\bf 0}$, it corresponds to the fused-lasso penalty \citep{Tibshirani05FusedLasso}, which leads to smoothness among neighboring regression coefficients. As we will show later, more complicated equality constraints can be incorporated with properly designed $\bV$ and $\bd$. On the other hand, the second regularization term $\rho \|\bW \bbeta - \be\|_+$ enforces regularization by inequality relations among regression coefficients. For large enough $\rho$, the minimizer $\bbeta(\rho)$ satisfies $\bW \bbeta(\rho) \le \be$. For instance, setting $\bW$ as the negative identity matrix and $\be = {\bf 0}$ encourages nonnegativity of the estimates, as required in nonnegative least squares problems \citep{LawsonHanson87LSBook}. In the isotonic regression \citep{RobertsonWrightDykstra88Book,SilvapullePranab05CSIBook}, the estimates have to be nondecreasing. This can be achieved by the regularization matrix
\begin{align*}
    \bW &= \left( \begin{array}{rrrrr}
    1 & -1 &  \\
     & \ddots & \ddots  \\
     & & 1 & -1
    \end{array}  \right)
\end{align*}
and $\be = {\bf 0}$. More complicated constraints that occur in shape-restricted regression and nonparametric regressions also can be incorporated as we demonstrate in later examples.

In certain applications, both equality and inequality regularizations are required. In that case,  as shown in Section \ref{sec:exact-penalty-method}, at a large but finite $\rho$, the minimizer $\bbeta(\rho)$ coincides with the solution to the following constrained optimization problem
\begin{eqnarray}
& \min& f(\bbeta)   \label{linconsobj}  \\
&\mbox{s.t.} &  \bV \bbeta = \bd \mbox{ and } \bW \bbeta \le \be. \nonumber
\end{eqnarray}
Consequently EPSODE solves the linearly constrained estimation problem (\ref{linconsobj}) as a by-product. In this case, path following commences from the unconstrained solution $\text{argmin} f(\bbeta)$ and ends at the constrained solution to (\ref{linconsobj}).

\subsection{Previous Work}

Several path algorithms have been devised for special cases of the general regularization problem (\ref{eqn:pathalgo-obj}). For example, the homotopy method \citep{OsbornePresnellTurlach00LassoAlgo} and the least angle regression (LARS) procedure \citep{EfronHastieIainTibshirani04LARS} handle lasso penalized least squares problem. The solution path generated is piecewise linear and illustrates the tradeoffs between goodness of fit and sparsity. \cite{RossetZhu07Path} give sufficient conditions for a solution path to be piecewise linear and expand its applications to a wider range of loss and penalty functions. Recently \cite{TibshiraniTaylor10GenLasso} devise a dual path algorithm for generalized $l_1$ penalized least squares problems, which is problem (\ref{eqn:pathalgo-obj}) with $f$ quadratic but without the second inequality regularization term. \citet{ZhouLange11LSPath} consider (\ref{eqn:pathalgo-obj}) in full generality for quadratic $f$. All these work concerns regularized linear regression for which the solution path is piecewise linear. Several attempts have been made to path following for regularized GLMs for which the solution path is no longer piecewise linear. \cite{ParkHastie07GLMLasso} propose a predictor-corrector approach to approximate the lasso path for GLMs. \cite{Wu10ODELasso} presents an ordinary differential equation-based path algorithm which delivers the {\it exact} path for lasso penalized GLMs. \cite{Friedman08GPS} derives an approximate path algorithm for any convex loss regularized by a separable, but not necessarily convex penalty. Here a penalty function is called separable if its Hessian matrix is diagonal. The separability restriction on the penalty term excludes many important problems encountered in real applications.

Our proposed approach generalizes previous work in several aspects. First, it works for any convex loss (or criterion) function. Second, it allows for any type of regularization in terms of linear functions of parameters, equality or inequality. Equality constrained regularizations include lasso, fused-lasso and generalized $l_1$ penalty for example. Inequality constrained regularizations are required in  shape-restricted regression and nonparametric log-concave density estimation. Last but not least, it is an {\it exact} path algorithm.

\subsection{A Motivating Example}

For illustration, we consider a merger and acquisition (M\&A) data set studied in \citep{FanMaityWangWu11ManA}. This data set constitutes $n=1,371$ US companies with a binary response variable indicating whether the company becomes a leveraged buyout (LBO) target ($y_i=1$) or not ($y_i=0$). Seven covariates (1. cash flow, 2. cash, 3. long term investment, 4. market to book ratio, 5. log market equity, 6. tax, 7. return on S\&P 500 index) are recorded for each company. There have been intensive studies on the effects of these factors on the probability of a company being a target for strategic mergers. Exploratory  analysis using linear logistic regression shows no significance in most covariates.

To explore the possibly nonlinear effects of these quantitative covariates, the varying-coefficient model \citep{HastieTibshirani93VaryingCoeff} can be adopted here. We discretize each predictor into, say, 10 bins and fit a logistic regression. The first bin of each predictor is used as the reference level and effect coding is applied to each discretized covariate. The circles (o) in Figure \ref{fig:MandA-estimates} denote the estimated coefficients for each bin of each predictor and hint at some interesting nonlinear effects. For instance, the chance of being an LBO target seems to monotonically decease with market-to-book ratio and be quadratic as a function of log market equity. Regularization can be utilized to borrow strength between neighboring bins and gear solution towards clearer patterns. To illustrate the flexibility of the regularization scheme (\ref{eqn:pathalgo-obj}), we apply cubic trend filtering to 5 covariates (cash flow, cash, long term investment, tax, return on S\&P 500 index), impose the monotonicity (non-increasing) constraint on the `market-to-book ratio' covariate, and enforce the concavity constraint on the `log market equity' covariate. This can be achieved by minimizing a regularized negative logistic log-likelihood of form
\begin{align*}
    - l(\bbeta_1,\ldots,\bbeta_7) + \rho \sum_{j \ne 4,5} \|\bV_j \bbeta_j\|_1 + \rho \sum_{j=4,5} \|\bW_j \bbeta_j\|_+,
\end{align*}
where $\bbeta_j$ is the vector of regression coefficients for the $j$-th discretized covariate. The matrices in the regularization terms are specified as
\begin{align*}
    \bV_j &= \left( \begin{array}{rrrrrrrr}
    -1 & 2 & -1 \\
    1 & -4 & 6 & -4 & 1 \\
      & 1 & -4 & 6 & -4 & 1 &   \\
      & & & \ddots & \ddots & \ddots  \\
    &  & 1 & -4 & 6 & -4 & 1 &  \\
      & & & & & -1 & 2 & -1
    \end{array}  \right) \text{ for } j=1,2,3,6,7, \\
    \bW_4 &= \left( \begin{array}{rrrrrr}
    -1 & 1 &  \\
     & -1 & 1 &  \\
     & & \ddots & \ddots  \\
    &  &  & -1 & 1 &  \\
      & & & & -1 & 1
    \end{array}  \right),   \mbox{ and }  \\
    \bW_5 &= \left( \begin{array}{rrrrrrr}
    1 & -2 & 1 \\
     & 1 & -2 & 1  \\
     & & \ddots & \ddots & \ddots  \\
    &  &  & 1 & -2 & 1 & \\
      & & & & 1 & -2 & 1
    \end{array}  \right).
\end{align*}
The equality constraint regularization matrix $\bV_j$, $j=1,2,3,6,7$, penalizes the fourth order finite differences between the bin estimates. Thus, as $\rho$ increases, the coefficient vectors of covariates 1-3,6-7 tend to be piecewise cubic with two ends being linear, mimicking the natural cubic spline. This is one example of the polynomial trend filtering \citep{KimKohBoyd09TrendFiltering,TibshiraniTaylor10GenLasso}. Similar to semi-parametric regressions, regularizations in polynomial trend filtering `let the data speak for themselves'. In contrast, the bandwidth selection in semi-parametric regressions is replaced by parameter tuning in regularizations. The number and locations of knots are automatically determined by tuning parameter which is chosen according to model selection criteria. In a similar fashion, the coefficient vector gradually becomes monotone for covariate `market-to-book ratio' and concave for covariate `log market equity'. In addition, with $\rho$ large enough, we recover the corresponding constrained solution, which are shown by the crosses (+) on solid lines in Figure \ref{fig:MandA-estimates}. As noted above, our exact path algorithm delivers the whole solution path bridging from the unconstrained estimates (denoted by o) to the constrained estimates (denoted by +). For example, the dotted lines in Figure \ref{fig:MandA-estimates} is a snapshot of the solution at $\rho=0.6539$. Availability of the whole solution path renders model selection along the path easy. For instance the regularization parameter $\rho$ can be chosen by minimizing the cross-validation error or other model selection criteria such as AIC, BIC, or $C_p$. Figure \ref{fig:MandA-path} displays the solution path and the AIC and BIC along the path. It shows that both criteria favor the fully regularized solution, namely the constrained estimates. The whole solution path is obtained within seconds on a laptop using a {\sc Matlab} implementation of EPSODE.

The patterns revealed by the regularized estimates match some existing finance theories. For instance, a company with low cash flow is unlikely to be an LBO target because low cash flow is hard to meet the heavy debt burden associated with the LBO. On the other hand, company carrying a high cash flow is likely to possess a new technology. It is risky to acquire such firms because it is hard to predict their profitability. The tax reason is obvious from the regularized estimates. The more tax the company is paying, the more tax benefits from an LBO. Log of market equity is a measure of company size. Smaller companies are unpredictable in their profitability and extremely large companies are unlikely to be an LBO target because LBOs are typically financed with a large proportion of external debts. Interested readers are referred to \citep{ShivdasaniWang09} and references therein for  related theories on LBO.

This illustrative example demonstrates the flexibility of our novel path algorithm. First, it can be applied to any convex loss function. In this example, the loss function is the negative log-likelihood of a logistic model. Second, it works for complicated regularizations like polynomial trend filtering (equality constraints), monotonicity constraint, and concavity constraint. More applications will be presented in Section \ref{sec:examples} to illustrate the potential of EPSODE.

\begin{figure}
\begin{center}
$$
\begin{array}{cc}
\includegraphics[width=6in]{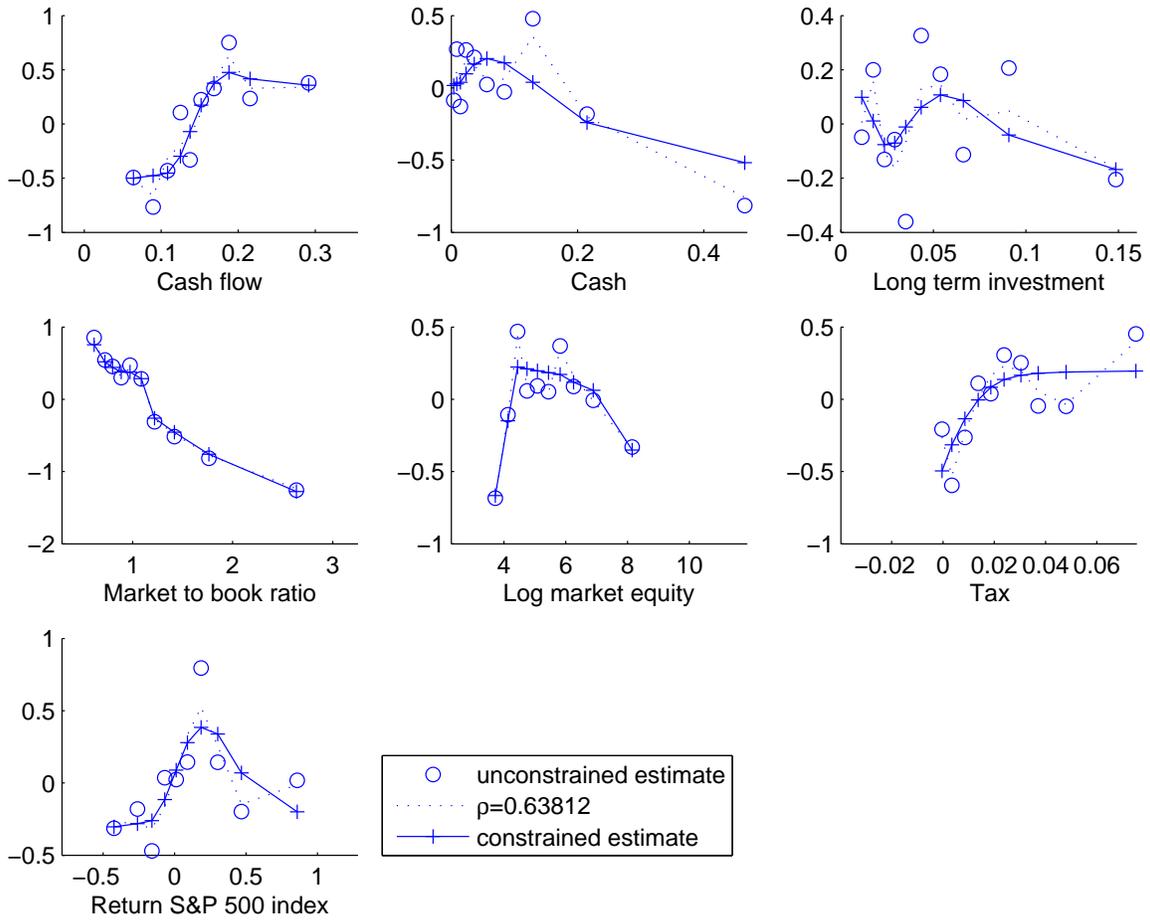}
\end{array}
$$
\caption{Snapshots of the path solution to the regularized logistic regression on the M\&A data set.}
\label{fig:MandA-estimates}
\end{center}
\end{figure}

\begin{figure}
\begin{center}
$$
\begin{array}{cc}
\includegraphics[width=2.5in]{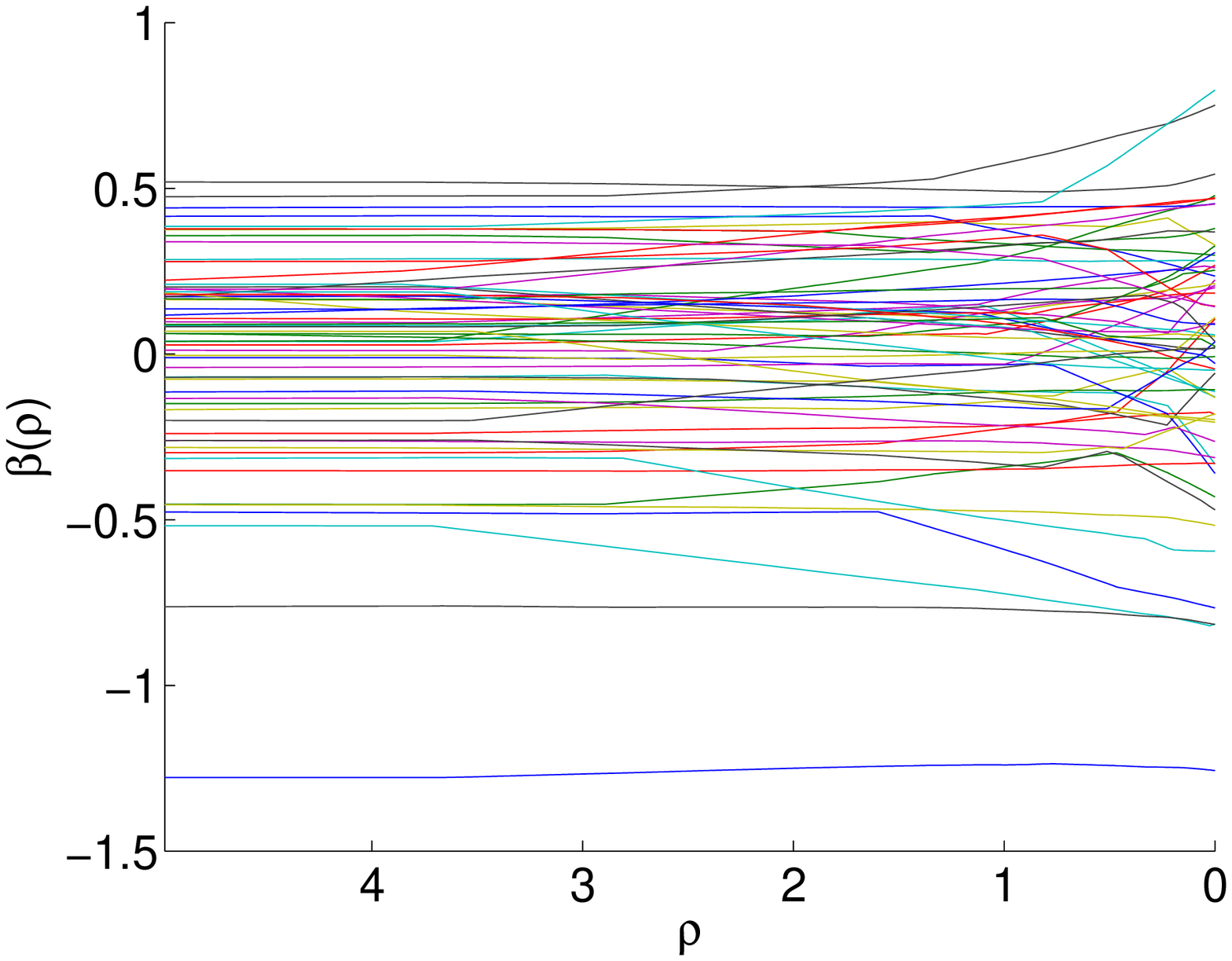} & \includegraphics[width=2.5in]{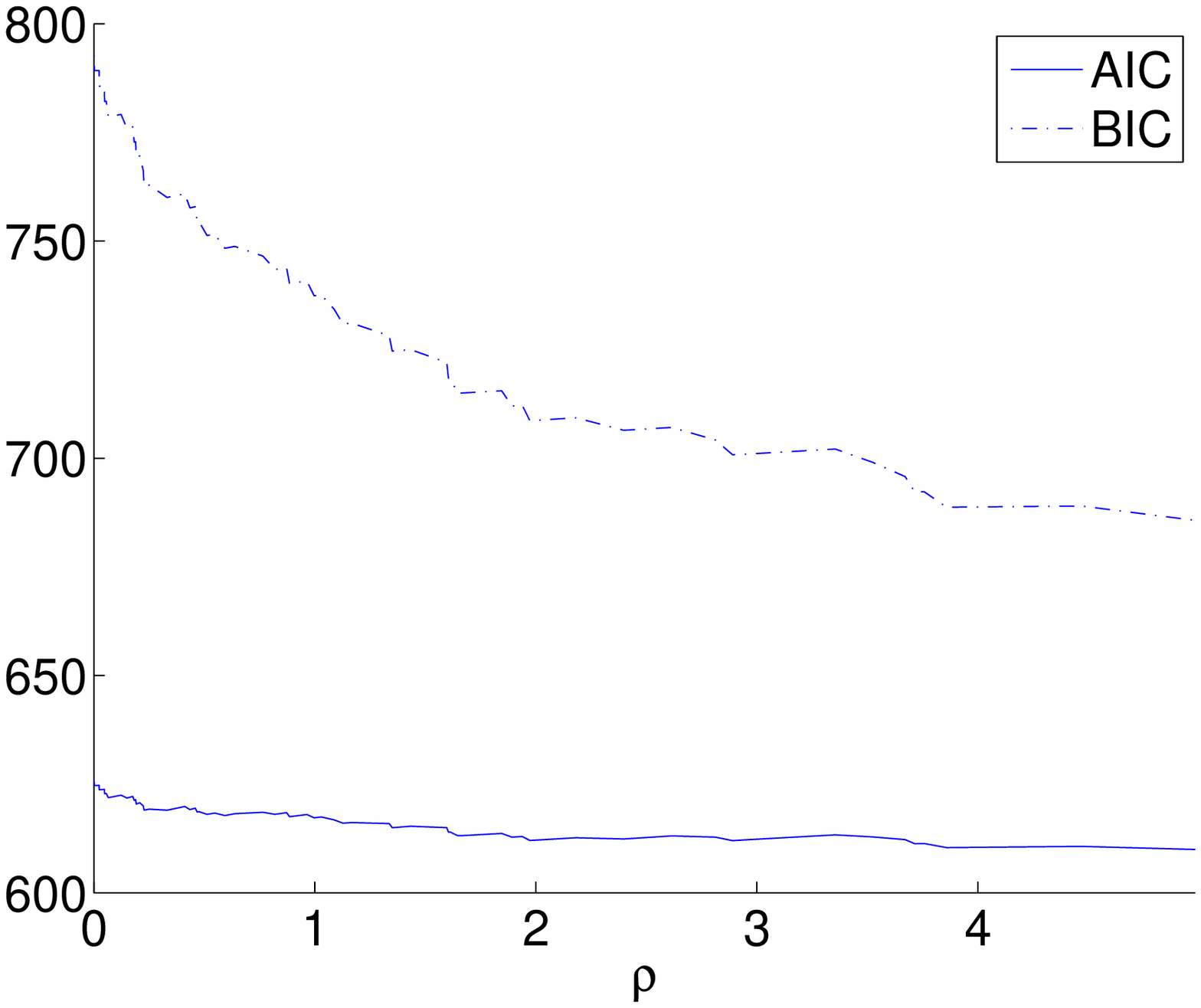}
\end{array}
$$
\caption{Solution and AIC/BIC paths of the regularized logistic regression on the M\&A data set.}
\label{fig:MandA-path}
\end{center}
\end{figure}

The rest of the paper is organized as follows. Section \ref{sec:exact-penalty-method} reviews the exact penalty method for optimization. Here the connections between constrained optimization and regularization in statistics are made clear. Section \ref{sec:path-algorithm} derives in detail the EPSODE algorithm for strictly convex loss function $f$. Its implementation via the sweep operator and ordinary differential equations are described in Section \ref{sec:implmentation}. An extension of EPSODE for $f$ convex but not necessarily strictly convex is discussed in Section \ref{sec:ext}. Section \ref{sec:dof} concerns model selection along the path. Section \ref{sec:examples} presents various applications of EPSODE. Finally, Section \ref{sec:conclusions} discusses the limitations of the path algorithm and hints at future generalizations.

\section{Exact Penalty Method for Convex Constrained Optimization}
\label{sec:exact-penalty-method}

Consider the convex program
\begin{eqnarray}
    &\min& f(\bx)  \nonumber \\
    &\text{s.t.}& g_i(\bx) = 0, 1 \le i \le r \label{convexprogram}\\
    & & h_j(\bx) \le 0, 1 \le j \le s, \nonumber
\end{eqnarray}
where the objective function $f$ is convex, equality constraint functions $g_i$ are affine, and the inequality constraint functions $h_j$ are convex. We further assume that $f$ and $h_j$ are smooth. Specifically we require that $f$ and $h_j$ are continuously twice differentiable. To fix notation, differential $df(\bx)$ is the row vector of partial derivatives of $f$ at $\bx$ and the gradient $\nabla f(\bx)$ is the transpose of $df(\bx)$. The Hessian matrix of $f(\cdot)$ is denoted by $d^2f(\bx)$.

Exact penalty method minimizes the function
\begin{eqnarray}
{\cal E}_{\rho}(\bx) & = & f(\bx)+\rho \sum_{i=1}^r |g_i(\bx)|+\rho \sum_{j=1}^s \max\{0,h_j(\bx)\}    \label{eqn:exact-pen-unconst}
\end{eqnarray}
for $\rho\geq0$. Classical results \citep[Theorems 6.9 and 7.21]{Ruszczynski06Book} state that for $\rho$ large enough, the solution to the  optimization problem (\ref{eqn:exact-pen-unconst}) coincides with the solution to the original constrained convex program (\ref{convexprogram}). This justifies the exact penalty method as one way to solve constrained optimization problems.

According to convex calculus \citep[Theorem 3.5]{Ruszczynski06Book}, the optimal point $\bx(\rho)$ of the function ${\cal E}_\rho(\bx)$ is characterized by the necessary and sufficient condition
\begin{eqnarray}
{\bf 0} & = & \nabla f(\bx) + \rho \sum_{i=1}^r s_i \nabla g_i(\bx) + \rho \sum_{j=1}^s t_j \nabla h_j(\bx) \label{path_stationary}
\end{eqnarray}
with coefficients satisfying
\begin{align}
   s_i \in \begin{cases}
   \{-1\} & g_i(\bx) < 0 \\
   [-1, 1] & g_i(\bx) = 0    \\
   \{1\} & g_i(\bx) > 0
   \end{cases}, \hspace{.5in} \mbox{ and } \hspace{0.5in} t_j \in \begin{cases}
   \{0\} & h_j(\bx) < 0 \\
   [0, 1] & h_j(\bx) = 0    \\
   \{1\} & h_j(\bx) > 0
   \end{cases}. \label{eqn:path-subgradient}
\end{align}
The sets defining possible values of $s_i$ and $t_j$ are the subdifferentials of the functions $|x|$ and $x_+ = \max\{x,0\}$. For path following to make sense, we require uniqueness and continuity of the solution $\bx(\rho)$ to (\ref{eqn:exact-pen-unconst}) as $\rho$ varies. The following lemma concerns the continuity of the solution path and is the foundation of our path algorithm.
\begin{lemma}
\label{lemma:continuity}
\begin{enumerate}
\item (Uniqueness) If ${\cal E}_\rho$ is strictly convex, then its minimizer $\bx(\rho)$ is unique.
\item (Continuity) If ${\cal E}_\rho$ is strictly convex and coercive over an open neighborhood of $\rho$, then the minimizer $\bx(\rho)$ is continuous at $\rho$.
\item (Continuity of $s_i$ and $t_j$) Furthermore, if the gradients $\{\nabla g_i(\bx): g_i(\bx)=0\} \cup \{\nabla h_j(\bx): h_j(\bx)=0\}$ of active constraints are linearly independent at the solution $\bx(\rho)$ over an open neighborhood of $\rho$, then the coefficient paths $s_i(\rho)$ and $t_j(\rho)$ are unique and continuous at $\rho$.
\end{enumerate}
\end{lemma}
\begin{proof} The uniqueness of minimum under strict convexity is well-known \citep{Ruszczynski06Book}. For continuity, suppose that the solution $\bx(\rho)$ is not continuous at $\rho$. Then there exists $\epsilon>0$ and a sequence $\rho_n \to \rho$ such that $\|\bx(\rho_n)-\bx(\rho)\|\ge \epsilon$ for all $n$. Since ${\cal E}_\rho$ is coercive, $\bx(\rho_n)$ is bounded and there exists a subsequence of $\bx(\rho_n)$ that converges to some point $\by$. Taking limits in the inequality ${\cal E}_{\rho_n}[\bx(\rho_n)] \le {\cal E}_{\rho_n}(\bx)$ shows that ${\cal E}_{\rho}(\by) \le {\cal E}_\rho(\bx)$ for all $\bx$, i.e., $\by = \bx(\rho)$. This contradicts with $\|\by - \bx(\rho)\| \ge \epsilon$. Therefore $\bx(\rho)$ is continuous at $\rho$. For the continuity of coefficients, under the linearly independence assumption, $s_i(\rho)$ and $t_j(\rho)$ can be uniquely solved by the stationarity condition (\ref{path_stationary}) given solution vector $\bx(\rho)$. Therefore continuity of $s_i(\rho)$ and $t_j(\rho)$ inherits from continuity of $\bx(\rho)$.
\end{proof}
\noindent
We remark that strict convexity only gives an easy-to-check sufficient condition for uniqueness and continuity; it is not a necessary condition. A convex but not strictly convex function can still have a unique minimum. The absolute value function $|x|$ is such an example.  When the loss function $f$ is strictly convex, then ${\cal E}_\rho$ is strictly convex for all $\rho \ge 0$ and by Lemma \ref{lemma:continuity} there exists a unique, continuous solution path $\{\bx(\rho): \rho \ge 0\}$.  In Section \ref{sec:path-algorithm} and \ref{sec:implmentation}, we derive the path algorithm assuming that $f$ is strictly convex. When $f$ is convex but not strictly convex, e.g., when $n<p$ in the least squares problems, the solutions at smaller $\rho$ may not be unique. In that case, it is still possible to obtain a solution path over the region of large $\rho$ where the minimum of ${\cal E}_\rho$ is unique. In Section \ref{sec:ext}, we extend EPSODE to the case $f$ is convex but may not be strictly convex. The third statement of Lemma \ref{lemma:continuity} implies that the active constraints ($g_i(\bx)=0$ or $h_j(\bx)=0$) with interior coefficients must stay active until the coefficients hit the end points of the permissible range, which in turn implies that the solution path is piecewise smooth. This allows us to develop a path following algorithm based on ODE.

\section{The Path Following Algorithm}
\label{sec:path-algorithm}

In this article, we specialize to the case where the constraint functions $g_i$ and $h_j$ are affine, i.e., the gradient vectors $\nabla g_i(\bx)$ and $\nabla h_j(\bx)$ are constant. This leads to the regularized optimization problem formulated as (\ref{eqn:pathalgo-obj}) by defining $g_i$ and $h_j$ as constraint residuals $g_i(\bx) = \bv_i^t \bx - d_i$ and $h_j(\bx) = \bw_j^t \bx - e_j$. In principle a similar path algorithm can be developed for the general convex program where the inequality constraint functions $h_j$ are relaxed to be convex. But that is beyond the scope of the current paper. In Sections \ref{sec:path-algorithm} and \ref{sec:implmentation}, we assume that the loss function $f$ is strictly convex. This assumption is relaxed in Section \ref{sec:ext}.

Our path following algorithm EPSODE works in a segment-by-segment fashion. Along the path we keep track of the following index sets determined by signs of constraint residuals
\begin{align}
   {\cal N}_{\text{E}} &= \{i: g_i(\bx) = \bv_i^t \bx - d_i < 0\}, \hspace{.5in} {\cal N}_{\text{I}} = \{j: h_j(\bx) = \bw_j^t \bx - e_j < 0\}  \nonumber  \\
   {\cal Z}_{\text{E}} &= \{i: g_i(\bx) = \bv_i^t \bx - d_i = 0\}, \hspace{.5in} {\cal Z}_{\text{I}} = \{j: h_j(\bx) = \bw_j^t \bx - e_j = 0\}  \label{eqn:set-config}  \\
   {\cal P}_{\text{E}} &= \{i: g_i(\bx) = \bv_i^t \bx - d_i > 0\}, \hspace{.5in} {\cal P}_{\text{I}} = \{j: h_j(\bx) = \bw_j^t \bx - e_j > 0\}  \nonumber.
\end{align}
Along each segment of the path, the set configuration is fixed. This is implied by the continuity of both the solution and coefficient paths established in  Lemma \ref{lemma:continuity}. Throughout this article, we call the constraints in ${\cal Z}_{\text{E}}$ or ${\cal Z}_{\text{I}}$ active and others inactive.

Next we derive the ODE for the solution $\bx(\rho)$ on a fixed segment. Suppose we are in the interior of a segment. Let $\bx(\rho)$ be the solution of (\ref{eqn:exact-pen-unconst}) indexed by the penalty parameter $\rho$ and $\bx(\rho+\Delta \rho)$ the solution when the penalty is increased by an infinitesimal amount $\Delta \rho > 0$. Then the difference $\Delta \bx(\rho) = \bx(\rho+\Delta \rho) - \bx(\rho)$ should minimize the increase in optimal objective value. That is, to the second order, $\Delta \bx$ is the solution to
\begin{eqnarray}
    &\min_{\Delta \bx}& {\cal E}_{\rho+\Delta\rho} (\bx + \Delta \bx) - {\cal E}_\rho(\bx) \nonumber \\
    &\approx& d f(\bx) \cdot \Delta \bx + \frac 12 \Delta \bx^t \cdot d^2f(\bx) \cdot \Delta \bx \nonumber  \\
    & & + (\rho+\Delta \rho) \cdot \left[ - \sum_{i \in {\cal N}_{\text{E}}} \bv_i + \sum_{i \in {\cal P}_{\text{E}}} \bv_i + \sum_{j \in {\cal P}_{\text{I}}} \bw_j \right] \cdot \Delta \bx \label{eqn:deltax} \\
    & & + \Delta \rho \cdot \left[- \sum_{i \in {\cal N}_{\text{E}}} g_i(\bx) + \sum_{i \in {\cal P}_{\text{E}}} g_i(\bx) + \sum_{j \in {\cal P}_{\text{I}}} h_j(\bx) \right]  \nonumber \\
    &\text{s.t.}& \bv_i^t \cdot \Delta \bx = 0, i \in {\cal Z}_{\text{E}}, \nonumber \\
    & & \bw_j^t \cdot \Delta \bx = 0, j \in {\cal Z}_{\text{I}}.  \nonumber
\end{eqnarray}
Note that the active constraints have to be kept active since the set configuration is fixed along this segment by Lemma \ref{lemma:continuity}. This is why we have these two sets of equality constraints. To ease notational burden, we define
\begin{eqnarray}
    \bH (\bx) &=& d^2f(\bx) \label{eqn:def-Hu}  \\
    \bu_{\bar{\cal Z}} &=& - \sum_{i \in {\cal N}_{\text{E}}} \bv_i + \sum_{i \in {\cal P}_{\text{E}}} \bv_i + \sum_{j \in {\cal P}_{\text{I}}} \bw_j. \nonumber
\end{eqnarray}
This leads to the corresponding Lagrange multiplier problem
\begin{eqnarray*}
    & & \left( \begin{array}{cc} \bH(\bx) & \bU^t_{{\cal Z}} \\ \bU_{{\cal Z}} & {\bf 0} \end{array} \right) \left( \begin{array}{c} \Delta \bx \\ \blambda_{{\cal Z}} \end{array} \right) = \left( \begin{array}{c} - \nabla f(\bx) - (\rho + \Delta \rho) \bu_{\bar{\cal Z}} \\ {\bf 0} \end{array} \right),
\end{eqnarray*}
where the rows of the matrix $\bU_{{\cal Z}}$ are the constant differentials, $\bv_i^t$, $i \in {\cal Z}_{\text{E}}$, and $\bw_j^t$, $j \in {\cal Z}_{\text{I}}(\bx)$, of the active constraint functions. Denoting the inverse of matrix as
\begin{eqnarray*}
    \left( \begin{array}{cc} \bH(\bx) & \bU^t_{{\cal Z}} \\ \bU_{{\cal Z}} & {\bf 0} \end{array} \right)^{-1} = \left( \begin{array}{cc} \bP(\bx) & \bQ(\bx) \\ \bQ^t(\bx) & \bR(\bx) \end{array} \right),
\end{eqnarray*}
where
\begin{eqnarray}
\bP(\bx) & = & \bH^{-1}(\bx) - \bH^{-1}(\bx) \bU_{{\cal Z}}^t \left[ \bU_{{\cal Z}} \bH^{-1}(\bx) \bU_{{\cal Z}}^t \right]^{-1} \bU_{{\cal Z}} \bH^{-1}(\bx) \nonumber \\
\bQ(\bx) & = & \bH^{-1}(\bx) \bU_{{\cal Z}}^t \left[ \bU_{{\cal Z}} \bH^{-1}(\bx) \bU_{{\cal Z}}^t \right]^{-1} \label{eqn:PQR} \\
\bR(\bx) & = & - \left[ \bU_{{\cal Z}} \bH^{-1}(\bx) \bU_{{\cal Z}}^t(\bx) \right]^{-1},    \nonumber
\end{eqnarray}
the solution of the difference vector $\Delta \bx$ is
\begin{eqnarray*}
    \Delta \bx &=& - \bP(\bx) [\nabla f(\bx) + (\rho + \Delta \rho) \bu_{\bar{\cal Z}}]   \\
    &=& - \bP(\bx) [\nabla f(\bx) + \rho \bu_{\bar{\cal Z}}(\bx) + \Delta \rho \cdot \bu_{\bar{\cal Z}}]    \\
    &=& - \bP(\bx) [- \rho \bU^t_{{\cal Z}} \br_{{\cal Z}} + \Delta \rho \cdot \bu_{\bar{\cal Z}}].
\end{eqnarray*}
Note $\bP(\bx) \bU^t_{{\cal Z}} = {\bf 0}$. Therefore $\Delta \bx = - \Delta \rho \cdot \bP(\bx) \bu_{\bar {\cal Z}}$. This gives the direction for the infinitesimal update of solution vector $\bx(\rho)$. Taking limit in $\Delta \rho$ leads to the following key result for developing the path algorithm.
\begin{proposition}
\label{prop:sol-ode}
Within interior of a path segment with set configuration (\ref{eqn:set-config}), the solution $\bx(\rho)$ satisfies an ordinary differential equation (ODE)
\begin{eqnarray}
    & & \frac{d\bx(\rho)}{d\rho} = - \bP(\bx) \bu_{\bar{\cal Z}} \label{eqn:sol-ode}
\end{eqnarray}
where the matrix $\bP(\bx)$ and vector $\bu_{\bar{\cal Z}}$ are defined by (\ref{eqn:PQR}) and (\ref{eqn:def-Hu}).
\end{proposition}
\noindent
Note that the right hand side of (\ref{eqn:sol-ode}) is a constant vector in $\bx$ when $f$ is quadratic and $g_i$ and $h_j$ are affine. Thus the corresponding solution path is piecewise linear. This recovers the case studied in \citep{ZhouLange11LSPath}. The differential equation (\ref{eqn:sol-ode}) holds on the current segment until one of two types of events happens: an inactive constraint becomes active or vice versa. The first type of event is easy to detect -- whenever a constraint function, $g_i(\bx)$, $i \in {\cal N}_{\text{E}} \cup {\cal P}_{\text{E}}$, or $h_j(\bx)$, $j \in {\cal N}_{\text{I}} \cup {\cal P}_{\text{I}}$, hits zero, we move that constraint to the active set ${\cal Z}_{\text{E}}$ or ${\cal Z}_{\text{I}}$ and start solving a new system of differential equations. To detect when the second type of event happens, we need to keep track of the coefficients $s_i(\bx)$ and $t_j(\bx)$ for active constraints. Whenever the coefficient of an active constraint hits the boundary of its permissible range in (\ref{eqn:path-subgradient}), the constraint has to be relaxed from being active in next segment. It turns out the coefficients for active constraints admit a simple representation in terms of current solution vector.
\begin{proposition}
\label{prop:coeff}
On a path segment with set configuration (\ref{eqn:set-config}), the coefficients $s_i$ and $t_j$ for active constraints are
\begin{eqnarray}
    & & \br_{{\cal Z}}(\rho) = \left( \begin{array}{c} \bs_{{\cal Z}_{\text{E}}}(\rho) \\ \bt_{{\cal Z}_{\text{I}}}(\rho) \end{array} \right)
    = - \bQ^t(\bx) \left[ \frac{1}{\rho} \nabla f(\bx) + \bu_{\bar {\cal Z}} \right]  \label{eqn:active-coeff}
\end{eqnarray}
where $\bx = \bx(\rho)$ is the solution at $\rho$ and the matrix $\bQ(\bx)$ is defined by (\ref{eqn:PQR}).
\end{proposition}
\begin{proof}
Stationarity condition (\ref{path_stationary}) implies
\begin{align*}
    \bU_{\cal Z}^t \br_{\cal Z} = - \frac{1}{\rho} \nabla f(\bx) - \bU_{\bar {\cal Z}}^t \br_{\bar {\cal Z}} = -\frac{1}{\rho} \nabla f(\bx) - \bu_{\bar {\cal Z}}.
\end{align*}
Multiplying both sides by $\bQ(\bx)$ gives (\ref{eqn:active-coeff}).
\end{proof}
\noindent
Given current solution vector $\bx(\rho)$, the coefficients of the active constraints are readily obtained from (\ref{eqn:active-coeff}). Once a coefficient hits the end points, we move that constraint from the active set to the inactive set that matches the end point being hit. In next section, we detail the implementation of the path algorithm.

\section{Implementation: ODE and Sweeping Operator}
\label{sec:implmentation}

\begin{algorithm}
\begin{algorithmic}
\STATE Initialize $\rho = 0$, $\bbeta(0) = \text{argmin} f(\bbeta)$ and its set configuration (\ref{eqn:set-config}).
\REPEAT
\STATE Solve ODE (\ref{eqn:sol-ode}) until an inactive constraint becomes active or the coefficient (\ref{eqn:active-coeff}) of an active constraint hits boundary.
\STATE Update the set configuration (\ref{eqn:set-config}).
\UNTIL{${\cal N}_{\text{E}} = {\cal P}_{\text{E}} = {\cal P}_{\text{I}} = \emptyset$}
\end{algorithmic}
\caption{EPSODE: Solution path for regularization problem (\ref{eqn:pathalgo-obj}) with strictly convex $f$.}
\label{algo:primal-path}
\end{algorithm}

Algorithm \ref{algo:primal-path} summarizes EPSODE based on Propositions \ref{prop:sol-ode} and \ref{prop:coeff}. It involves solving ODEs segment by segment and is extremely simple to implement using softwares with a reliable ODE solver such as the {\tt ode45} function in {\tt Matlab} and the {\tt deSolve} package \citep{Karline10RdeSolve} in {\tt R}. There has been extensive research in applied mathematics on numerical methods for solving ODEs, notably the Runge-Kutta, Richardson extrapolation and predictor-corrector methods. Some  path following algorithms developed for specific statistical problems \citep{ParkHastie07GLMLasso,Friedman08GPS} turn out to be approximate methods for solving the corresponding ODE. \cite{Wu10ODELasso} first explicitly uses ODE to derive an exact solution path for the lasso penalized GLM. The connection of path following to ODE relieves statisticians from the burden of developing specific path algorithms for a variety of regularization problems.

Any ODE solver repeatedly evaluates the derivative. Suppose the number of parameters is $p$. Computation of the matrix-vector multiplications in (\ref{eqn:sol-ode}) and (\ref{eqn:active-coeff}) has computation cost of order $O(p^2)+O(p|{\cal Z}|)+O(|{\cal Z}|^3)$ if the inverse $H^{-1}$ of Hessian matrix of loss function $f$ is readily available, where ${\cal Z}={\cal Z}_{\text{E}}\cup {\cal Z}_{\text{I}}$ and $|{\cal Z}|$ denotes its cardinality. Otherwise the computation cost is $O(p^3)+O(p|{\cal Z}|)+O(|{\cal Z}|^3)$.

An alternative implementation avoids repeated matrix inversions by solving an ODE for the matrices $\bP$, $\bQ$ and $\bR$ themselves. The computations can be conveniently organized around the classical sweep and inverse sweep operators of regression analysis \citep{Dempster69Book,Goodnight79Sweep,Jennrich77Stepwisereg,LittleRubin02Book,Lange10NumAnalBook}. Suppose $\bA$ is an $m \times m$ symmetric matrix. Sweeping on the $k$th diagonal entry
$a_{kk} \ne 0$ of $\bA$ yields a new symmetric matrix $\widehat{\bA}$ with entries
\begin{eqnarray*}
    \hat{a}_{kk} &=& - \frac{1}{a_{kk}}, \\
    \hat{a}_{ik} &=& \frac{a_{ik}}{a_{kk}}, \quad i \ne k  \\
    \hat{a}_{kj} &=& \frac{a_{kj}}{a_{kk}}, \quad j \ne k   \\
    \hat{a}_{ij} &=& a_{ij} - \frac{a_{ik}a_{kj}}{a_{kk}}, \quad i,j \ne k .
\end{eqnarray*}
These arithmetic operations can be undone by inverse sweeping on the same diagonal entry. Inverse sweeping on the $k$th diagonal entry sends the symmetric matrix $\bA$ into the symmetric matrix $\check{\bA}$ with entries
\begin{eqnarray*}
    \check{a}_{kk} &=& - \frac{1}{a_{kk}}, \\
    \check{a}_{ik} &=&  - \frac{a_{ik}}{a_{kk}}, \quad i \ne  k\\
    \check{a}_{kj} &=& - \frac{a_{kj}}{a_{kk}},  \quad j \ne k\\
    \check{a}_{ij} &=& a_{ij} - \frac{a_{ik}a_{kj}}{a_{kk}}, \quad i,j \ne k.
\end{eqnarray*}
Both sweeping and inverse sweeping preserve symmetry. Thus, all operations can be carried out on either the lower or upper triangle of $\bA$ alone, saving both computational time and storage. When several sweeps or inverse sweeps are performed, their order is irrelevant.

At beginning ($\rho=0$) of the path following, we initialize a sweeping tableau as
\begin{eqnarray*}
   \left( \begin{array}{c|c}
   \bH^{-1}(\bx) & \bH^{-1}(\bx) \bU^t \\
   \hline
   * & \bU \bH^{-1}(\bx) \bU^t
   \end{array} \right),
\end{eqnarray*}
where the matrix $\bU \in \mathbb{R}^{(r+s)\times p}$ holds all constraint differentials $\bv_i^t$ and $\bw_j^t$ in rows. Further sweeping of diagonal entries corresponding to the active constraints yields
\begin{eqnarray}
   \left( \begin{array}{c|cc}
     \bP(\bx) & \bQ(\bx) & \bP(\bx) \bU_{\bar{\cal Z}}^t \\       \hline
   * & \bR(\bx) &  \bQ^t(\bx) \bU_{\bar{\cal Z}}^t \\
   * & * & \bU_{\bar{\cal Z}} \bP(\bx) \bU_{\bar {\cal Z}}^t
   \end{array} \right). \label{eqn:sweep-tableau}
\end{eqnarray}
Here we conveniently organized the columns of the swept active constraints before those of un-swept ones. In practice the sweep tableau is not necessary as in (\ref{eqn:sweep-tableau}) and it is enough to keep an indicator vector recording which columns are swept. The key elements for the path algorithm magically appear in the sweep tableau (\ref{eqn:sweep-tableau})
\begin{eqnarray*}
    \frac{d\bx(\rho)}{d\rho} &=& - \bP(\bx) \bU_{\bar{\cal Z}}^t \br_{\bar{\cal Z}}  \\
    \br_{{\cal Z}}(\rho) &=& - \bQ^t(\bx) \bU_{\bar{\cal Z}}^t \br_{\bar{\cal Z}} - \frac{1}{\rho} \bQ^t(\bx) \nabla f(\bx).
\end{eqnarray*}
Therefore path following procedure only involves solving ODE for the whole sweep tableau (\ref{eqn:sweep-tableau}) with sweeping or inverse sweeping at kinks between successive segments. For this purpose we derive the ODE for the sweep tableau (\ref{eqn:sweep-tableau}). We adopt the convenient notations in \citep{MagnusNeudecker99MatrixBook}. For a matrix function $F(\bX): \mathbb{R}^{n\times q} \to \mathbb{R}^{m \times p}$,
\begin{eqnarray*}
    DF(\bX) = \frac{\partial \text{vec} F(\bX)}{\partial (\text{vec} \bX)^t}
\end{eqnarray*}
denotes the $mp \times nq$ Jacobian matrix. For example, Proposition \ref{prop:sol-ode} states $D \bx(\rho) = - \bP(\bx) \bu_{\bar {\cal Z}}$.
\begin{proposition}[ODE for Sweep Tableau]
On a segment of path with fixed set configuration, the matrices $\bP(\rho)$, $\bQ(\rho)$ and $\bR(\rho)$ satisfy the ordinary differential equations (ODE)
\begin{align*}
    D \bP(\rho) &= [\bP(\bx) \otimes \bP(\bx)] \cdot [D \bH(\bx)] \cdot \bP(\bx) \bu_{\bar {\cal Z}} \\
    D \bQ(\rho) &= [\bQ^t(\bx) \otimes \bP(\bx)] \cdot [D \bH(\bx)] \cdot \bP(\bx) \bu_{\bar {\cal Z}}  \\
    D \bR(\rho) &=[\bQ^t(\bx) \otimes \bQ^t(\bx)] \cdot [D \bH(\bx)] \cdot \bP(\bx) \bu_{\bar {\cal Z}}.
\end{align*}
\end{proposition}
\begin{proof}
First consider
\begin{eqnarray*}
    \bR(\bx) = - \left[ \bU_{{\cal Z}} \bH^{-1}(\bx) \bU_{{\cal Z}}^t \right]^{-1} = -\bM^{-1}(\bx).
\end{eqnarray*}
By chain rule \cite[p91]{MagnusNeudecker99MatrixBook},
\begin{eqnarray*}
    D \bR(\rho)
    &=& D \bR(\bM) \cdot D \bM(\bH) \cdot D \bH(\bx) \cdot D \bx(\rho) \\
    &=& [\bR(\bx) \otimes \bR(\bx)] \cdot D \bM(\bH) \cdot D \bH(\bx) \cdot D \bx(\rho) \\
    &=& - [\bR(\bx) \otimes \bR(\bx)] \cdot \{ [\bU_{\cal Z} \bH^{-1}(\bx) \otimes \bU_{\cal Z} \bH^{-1}(\bx)] \cdot D \bH(\bx) \} \cdot D\bx(\rho) \\
    &=& [\bQ^t(\bx) \otimes \bQ^t(\bx)] \cdot [D \bH(\bx)] \cdot \bP(\bx) \bu_{\bar {\cal Z}}.
\end{eqnarray*}
Similar calculations yield formula for $\bQ(\bx) = - \bH^{-1}(\bx) \bU_{{\cal Z}}^t \bR(\bx)$ and $\bP(\bx) = \bH^{-1}(\bx) - \bQ(\bx) \bU_{\cal Z} \bH^{-1}(\bx)$.
\end{proof}

Solving ODE for these matrices requires the $p^2$-by-$p$ Jacobian matrix of the Hessian matrix $\bH(\bx) = d^2f(\bx)$,
\begin{eqnarray*}
    D \bH(\bx) = \frac{\partial [\text{vec} \bH(\bx)]}{\partial \text{vec} (\bx)^t} = \frac{\partial \text{vec}[df^2(\bx)]}{\partial \text{vec}(\bx)^t},
\end{eqnarray*}
which we provide for each example in Section \ref{sec:examples} for convenience. When the number of parameter $p$ is large, $D\bH$ is a large matrix. However there is no need to compute and store $D\bH$ and we are only required to compute the matrix vector multiplication $D\bH \cdot \bv$ for any vector $\bv$. In light of the useful identity $(\bB^t \otimes \bA) \text{vec}(\bC) = \text{vec}(\bA \bC \bB)$, evaluating the derivative for the whole tableau only involves multiplying three matrices and incurs computational cost $O(p^3) + O(p^2 |{\cal Z}|) + O(p|{\cal Z}|^2)$.

Although we have presented the path algorithm as moving from $\rho=0$ to large $\rho$, it can be applied in either direction. Lasso and fused-lasso usually start from the constrained solution, while in presence of general equality constraints, e.g., polynomial trend filtering, and/or inequality constraints, the constrained solution is not readily available and the path algorithm must be initiated at $\rho=0$.

\section{Extension of EPSODE}
\label{sec:ext}

So far we have assumed strictly convexity of the loss function $f$. This unfortunately excludes many interesting applications, especially $p>n$ case of the regression problems. In this section we briefly indicate an extension of EPSODE to the case $f$ is convex but not necessarily strictly convex. In the proof of Proposition \ref{prop:sol-ode}, the infinitesimal change of solution $\Delta \bx$ is derived via minimizing the equality-constrained quadratic program (\ref{eqn:deltax}), the solution to which requires inverse of Hessian $\bH^{-1}$ and thus strict convexity of $f$. Alternatively we may solve (\ref{eqn:deltax}) via reparameterization. Let $\bU_{\cal Z}$ hold the active constraint vectors and $\bY \in \mathbb{R}^{p \times (p-|{\cal Z}|)}$ be a null space matrix of $\bU_{\cal Z}$, i.e., the columns of $\bY$ are orthogonal to the rows of $\bU_{\cal Z}$. Then the infinitesimal change can be represented as $\Delta \bx = \bY \Delta \by$ for some vector $\Delta \by \in \mathbb{R}^{p-|{\cal Z}|}$. Under this reparameterization, the quadratic program (\ref{eqn:deltax}) is equivalent to
\begin{align*}
    \min_{\Delta \by}   \, \frac 12 \Delta \by^t [\bY^t \bH(\bx) \bY] \Delta \by + [ df(\bx) + (\rho + \Delta \rho) \bu_{\bar {\cal Z}}^t]\bY \cdot \Delta \by
\end{align*}
with explicit solution
\begin{align*}
    \Delta \by = - [\bY^t \bH(\bx) \bY]^{-1} \bY^t [\nabla f(\bx)+(\rho+\Delta \rho) \bu_{\bar {\cal Z}}].
\end{align*}
Hence the infinitesimal change in $\bx(\rho)$ is
\begin{align*}
    \Delta \bx &= - \bY [\bY^t \bH(\bx) \bY]^{-1} \bY^t [\nabla f(\bx)+(\rho+\Delta \rho)\bu_{\bar {\cal Z}}] \\
    &= - \Delta \rho \cdot \bY [\bY^t \bH(\bx) \bY]^{-1}\bY^t \bu_{\bar {\cal Z}}].
\end{align*}
Again taking limit gives the following result in parallel to Proposition \ref{prop:sol-ode}.
\begin{proposition}
\label{prop:sol-ode-alt}
Within interior of a path segment with set configuration (\ref{eqn:set-config}), the solution $\bx(\rho)$ satisfies an ordinary differential equation (ODE)
\begin{eqnarray}
    & & \frac{d\bx(\rho)}{d\rho} = - \bY [\bY^t \bH(\bx) \bY]^{-1}\bY^t \bu_{\bar{\cal Z}} \label{eqn:sol-ode-alt}
\end{eqnarray}
where $\bY$ is a null space matrix of $\bU_{\cal Z}$.
\end{proposition}
\noindent
An advantage of (\ref{eqn:sol-ode-alt}) is that only non-singularity of the matrix $\bY^t \bH(\bx) \bY$ is required which is much weaker than the non-singularity of $\bH$. The computational cost of calculating the derivative in (\ref{eqn:sol-ode-alt}) is $O((p-|{\cal Z}|)^3) + O(p(p-|{\cal Z}|))$, which is more efficient than (\ref{eqn:sol-ode}) when $p-|{\cal Z}|$ is small. However it requires the null space matrix $\bY$, which is nonunique and may be expensive to compute. Fortunately the null space matrix $\bY$ is constant over each path segment and in practice can be calculated by QR decomposition of the active constraint matrix $\bU_{\cal Z}$. At each kink either one constraint leaves ${\cal Z}$ or one enters ${\cal Z}$. Therefore $\bY$ can be sequentially updated \citep{LawsonHanson87LSBook} and need not to be calculated anew for each segment. Which version of (\ref{eqn:sol-ode}) and (\ref{eqn:sol-ode-alt}) to use depends on specific application. When the loss function $f$ is not strictly convex, e.g., $p>n$ case in regression analysis, only (\ref{eqn:sol-ode-alt}) applies.  Interested readers are referred to \citep{NocedalWright06Book} for a similar dilemma in optimization methods.

\section{Model Selection Along the Path}
\label{sec:dof}

In applications such as penalized GLMs,  the tuning parameter $\rho$ in the regularization problem (\ref{eqn:pathalgo-obj}) is chosen by a model selection criterion such as AIC, BIC, $C_p$, or cross-validation. The cross validation errors can be readily computed using the solution path output by EPSODE. Yet the AIC, BIC, and $C_p$ criteria require an estimate of the degrees of freedom of estimate $\bbeta(\rho)$. Specifically AIC and BIC are defined by
\begin{align*}
    \text{AIC} &= - \ell (\bbeta(\rho)) + \text{df}(\bbeta(\rho)) \\
    \text{BIC} &= - \ell (\bbeta(\rho)) + \frac{\log n}{2} \text{df}(\bbeta(\rho)),
\end{align*}
where $-\ell(\cdot)$ denotes the negative log-likelihood and $\text{df}(\bbeta_\rho)$ is the degrees of freedom for estimate $\bbeta_\rho$. We propose to use
\begin{align}
    \text{df}(\bbeta(\rho)) = p - |{\cal Z}_{\text{E}} \cup {\cal Z}_{\text{I}}|  \label{eqn:dof}
\end{align}
as a measure of the degrees of freedom under GLMs. It is previously shown that (\ref{eqn:dof}) is an unbiased estimate of the degrees of freedom for lasso penalized least squares \citep{EfronHastieIainTibshirani04LARS,ZouHastieTibshirani07LassoDF}, generalized lasso penalized least squares \citep{TibshiraniTaylor10GenLasso}, and the least squares version of the regularized problem (\ref{eqn:pathalgo-obj}) \citep{ZhouLange11LSPath}. Using the same degrees of freedom formula (\ref{eqn:dof}) for GLMs is justified by the local approximation of GLM loglikelihood by weighted least squares. See \citep{ParkHastie07GLMLasso} for details.

\section{Applications}
\label{sec:examples}

In this section, we collect some representative regularized or constrained estimation problems and demonstrate how they can be solved by path following. For all applications, we list the first three derivatives of the loss function $f$ in (\ref{eqn:pathalgo-obj}). In fact, the third derivative is only needed when implementing by solving the ODE for the sweep tableau.

\subsection{GLMs and Quasi-Likelihoods with Generalized $l_1$ Regularizations}

The generalized linear model (GLM) deals with exponential families in which the sufficient statistics is $Y$ and the conditional mean $\mu$ of $Y$ completely determines its distribution. Conditional on the covariate vector $\bx \in \mathbb{R}^{p}$, the response variable $y$ is modeled as
\begin{align}
    p(y|\bx; \bbeta, \sigma) \propto \exp \left\{ \frac{y \langle \bx, \bbeta \rangle - \psi(\langle\bx, \bbeta \rangle)}{c(\sigma)} \right\}, \label{eqn:glm}
\end{align}
where the scalar $\sigma > 0$ is a fixed and known scale parameter and the vector $\bbeta$ is the parameters to be estimated. The function $\psi: \mathbb{R} \mapsto \mathbb{R}$ is the link function. When $y \in \mathbb{R}$, $\psi(u) = u^2/2$ and $c(\sigma) = \sigma^2$, (\ref{eqn:glm}) is the {\it normal regression model}. When $y \in \{0,1\}$, $\psi(u) = \ln (1+ \exp(u))$ and $c(\sigma)=1$, (\ref{eqn:glm}) is the {\it logistic regression model}. When $y \in \mathbb{N}$, $\psi(u) = \exp(u)$, and $c(\sigma)=1$, (\ref{eqn:glm}) is the {\it Poisson regression model}.

The quasi-likelihoods generalize GLM without assuming a specific distribution form of $Y$. Instead only a function relation between the conditional means $\mu_i$ and variances $\sigma_i^2$,  $\sigma_i^2= V(\mu_i)$ for some variance function $V$, is needed. Then the integral
\begin{align*}
    Q(\mu, y) = \int_y^\mu \frac{y-t}{\sigma^2 V(t)} \, dt
\end{align*}
behaves like a log-likelihood function under mild conditions and is called the quasi-likelihood. The quasi-likelihood includes GLMs as special cases with appropriately chosen variance function $V(\cdot)$.  Readers are referred to the classical text \cite[Table 9.1]{McCullaghNelder83GLMBook} for the commonly used quasi-likelihoods. By slightly abusing our notation, we assume a known link function between the conditional mean $\mu_i$ and linear predictor $\bx_i^T\bbeta$, $\mu=\mu(\bx_i^T\bbeta)$  and denote $Q_i(\bbeta)=Q(\mu(\bx_i^T\bbeta), y_i)$. Then the quasi-likelihood with generalized $l_1$ regularization takes the form
\begin{eqnarray}
    & - Q(\bbeta) + \rho \|\bV \bbeta - \bd\|_1 = - \sum_{i=1}^n Q_i(\bbeta) + \rho \|\bV \bbeta - \bd\|_1, \label{eqn:quasiL-genl1}
\end{eqnarray}
which is a special case of the general form (\ref{eqn:pathalgo-obj}). Specific choices of the regularization matrix $\bV$ and constant vector $\bd$ lead to lasso, fused-lasso, trend filtering, and many other applications.

For the path algorithm, we require the first two or three derivatives of the complete quasi-likelihood. Denoting $\boldsymbol{\eta} = \bX \bbeta$ with $\bX=(\bx_1^t, \bx_2^t, \cdots, \bx_n^t)^t$, we have
\begin{eqnarray}
    \nabla Q(\bbeta) &=& [D \bmu(\boldsymbol{\eta})]^t \bV^{-1} (\by - \bmu) / \sigma^2 = \bX^t [D \bmu(\boldsymbol{\eta})] \bV^{-1} (\by - \bmu) / \sigma^2,   \nonumber \\
    \bH(\bbeta) = d^2Q(\bbeta) &=& [(\by-\bmu)^t \bV^{-1} \otimes \bX^t] \cdot D^2 \bmu(\boldsymbol{\eta}) \cdot \bX /\sigma^2,  \label{eqn:GLM-derivatives} \\
    D\bH(\bbeta) = d^3Q(\bbeta) &=& [\bX^t \otimes (\by-\bmu)^t \bV^{-1} \otimes \bX^t] \cdot D^3 \bmu(\boldsymbol{\eta}) \cdot \bX /\sigma^2,    \nonumber
\end{eqnarray}
where $\bV$ is a $n$-by-$n$ diagonal matrix with diagonal entries $V(\mu(\bx_i^t\by))$, $D \bmu(\boldsymbol{\eta})$ is a $n$-by-$n$ diagonal matrix with diagonal entries $\mu'(\bx_i^t\bbeta)$, $D^2 \bmu(\boldsymbol{\eta})$ is a $n^2$-by-$n$ matrix with $(n(i-1)+i,i)$ entry equal to $\mu''(\bx_i^t\bbeta)$ for $i=1,\ldots,n$ and 0 otherwise, and $D^3 \bmu(\boldsymbol{\eta})$ is a $n^3$-by-$n$ matrix with $(n^2(i-1)+n(i-1)+i,i)$ entry equal to $\mu'''(\bx_i^t\bbeta)$ for $i=1,\ldots,n$ and 0 otherwise. Note for GLM with canonical link, these formulas simplify \citep[Section 4.6.4]{Agresti02Book}.

The most widely used $l_1$ regularization is the lasso penalty which imposes sparsity on the regression coefficients. For numerical demonstration, we revisit the M\&A example introduced in Section \ref{sec:intro} without discretizing  each predictor. We standardize each predictor first and consider the lasso penalized linear logistic regression model.  Figure \ref{fig:MandA-lassopath} shows the lasso solution path for each standardized predictor in the left panel and corresponding AIC and BIC scores in the right panel. The order at which predictors enter the model matches the more detailed patterns revealed by the varying coefficient model in Figure \ref{fig:MandA-estimates}. The almost monotone effects of the predictors `market-to-book ratio', `cash flow', 'cash', and 'tax' can be captured by the usual linear logistic regression and these covariates are picked up by lasso first. The nonlinear effects shown in the other predictors are likely to be missed by the linear logistic regression. For instance, the quadratic effects of `log market equity' shown in the regularized estimates in Figure \ref{fig:MandA-estimates} are missed by both AIC and BIC criteria.

\begin{figure}
$$
\begin{array}{cc}
\includegraphics[width=2.5in]{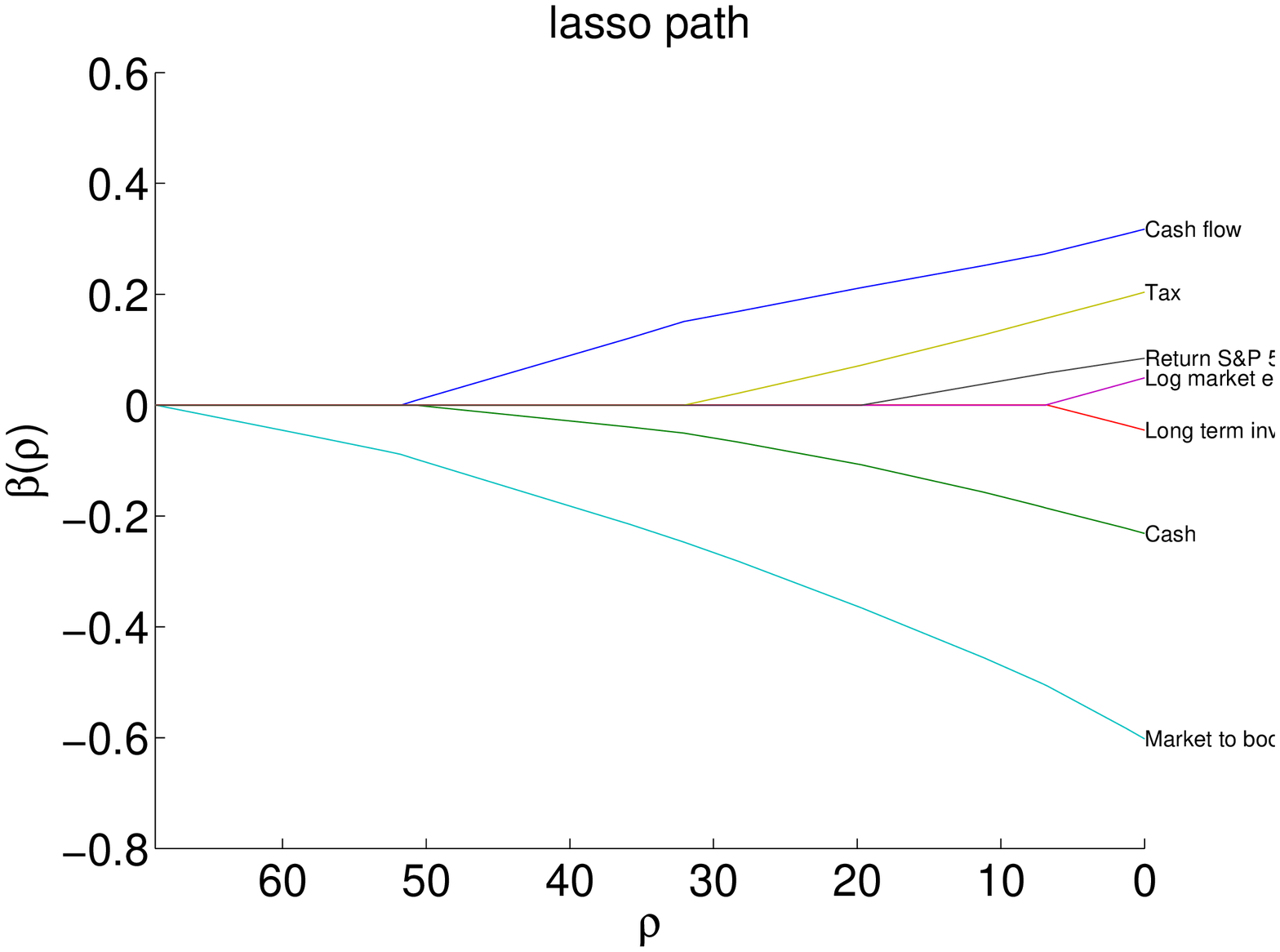} & \includegraphics[width=2.5in]{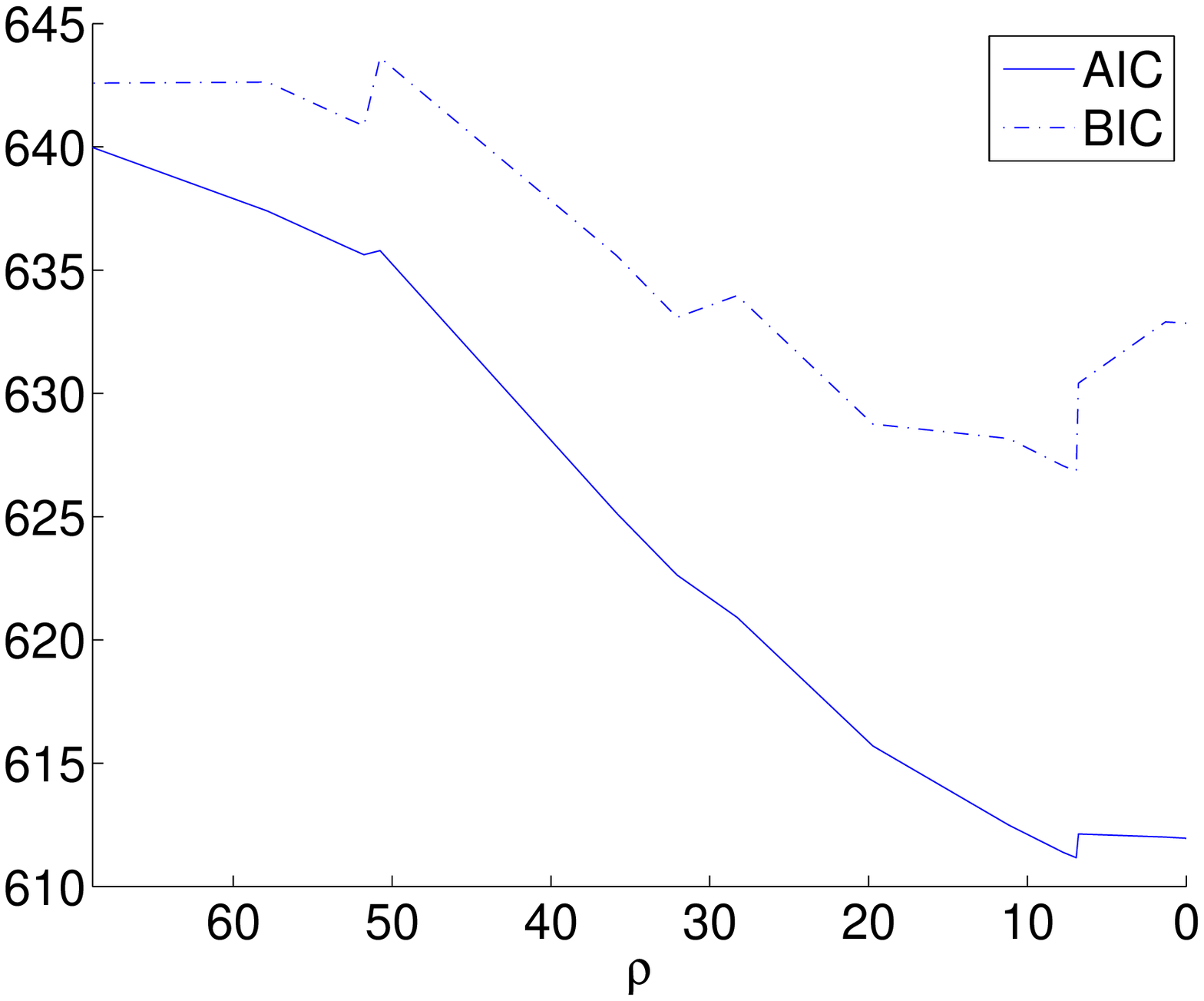}
\end{array}
$$
\caption{M\&A example revisited. Lasso solution path on the seven standardized predictors.}
\label{fig:MandA-lassopath}
\end{figure}

All the generalized lasso problems studied in \citep{TibshiraniTaylor10GenLasso} for Gaussian linear regression naturally generalize to GLMs or quasi-likelihoods and are subject to the EPSODE path algorithm. This leads to applications to lasso or fused-lasso penalized GLMs, outlier detections, trend filtering, and image restoration for GLMs. For instance, cubic trend filtering is performed on five predictors of the M\&A example in Section 1. The graph-guided penalized linear regression proposed in \citep{ChenLinKimCarbonellXing10ProxGrad} can also be generalized to GLMs or quasi-likelihoods. Suppose each node $i$ of a graph is assigned a regression coefficient $\beta_i$.  In graph penalized regression, the objective function takes the form
\begin{eqnarray}
    - \ell(\bbeta) + \lambda_{\text{G}} \sum_{i \sim j} \left| \frac{\beta_i}{\sqrt{d_i}} - \text{sgn}(r_{ij})\frac{\beta_j}{\sqrt{d_j}} \right| + \lambda_{\text{L}} \sum_j |\beta_j|, \label{graph_objective}
\end{eqnarray}
where the set of neighboring pairs $i \sim j$ define the graph, $d_i$ is the degree of node $i$, and $r_{ij}$ is the correlation coefficient between $i$ and $j$. This is simply a special case of (\ref{eqn:quasiL-genl1}) when the ratio $\lambda_{\text{G}}/\lambda_{\text{L}}$ is fixed.

\subsection{Shape-Restricted Regressions}

Order-constrained regression has been an important modeling tool \citep{RobertsonWrightDykstra88Book,SilvapullePranab05CSIBook}. If $\bbeta$ denotes the parameter vector, monotone regression imposes {\it isotone} constraints $\beta_1 \le \beta_2 \le  \cdots \le \beta_p$ or {\it antitone} constraints $\beta_1 \ge \beta_2 \ge  \cdots \ge \beta_p$. In {\it partially ordered regression}, subsets of the parameters are subject to isotone or antitone constraints.  In some other problems it is sensible to impose {\it convex} or {\it concave} constraints.
Note that if locations of regression parameters  are at irregularly spaced time points $t_1 \le t_2 \le \cdots \le t_p$, convexity translates into the constraints
\begin{eqnarray*}
\frac{\beta_{i+2}-\beta_{i+1}}{t_{i+2}-t_{i+1}} \ge \frac{\beta_{i+1}-\beta_i}{t_{i+1}-t_i}
\end{eqnarray*}
for $1 \le i \le p-2$. When the time intervals are uniform, the constraints simplify to  $\beta_{i+2} - \beta_{i+1} \ge \beta_{i+1} - \beta_{i}$, $i=1, 2, \cdots, p-1$. Concavity translates into the opposite set of inequalities.

Most of previous work has focused on the linear regression problems because of the computational and theoretical complexities in the generalized linear model setting. The recent work \citep{Rufibach10GLMOrder} proposes an active set algorithm for GLMs with order constraints. The EPSODE algorithm conveniently provides a solution to the linearly constrained estimation problem (\ref{linconsobj}). The relevant derivatives of loss function are listed in (\ref{eqn:GLM-derivatives}). It is noteworthy that EPSODE not only provides the constrained estimate but also the whole path bridging the unconstrained estimate to the constrained solution. Availability of the whole solution path renders model selection between the two extremes simple.

In the illustrative M\&A example of Section \ref{sec:intro}, the bin predictors for the `market-to-book ratio' are regularized by the antitone constraint and those for the `log market equity' covariate by the concavity constraint.

\subsection{Gaussian Graphical Models}

In recent years several authors \citep{Friedman08GLasso, Yuan08GraphLasso} proposed to estimate the sparse undirected graphical model by using lasso regularizations to the log-likelihood function of the precision matrix, the inverse of the variance-covariance matrix. Given an observed variance-covariance matrix $\hat \bSigma\in R^{p\times p}$, the negative log-likelihood of the precision matrix $\bOmega=\bSigma^{-1}$ under normal assumption is
\begin{align}
    f(\bOmega) = - \log \det \bOmega + \text{tr}(\hat \bSigma \bOmega)    \label{eqn:graph-obj}
\end{align}
with the MLE solution $\hat \bSigma^{-1}$ when $\hat\bSigma$ is non-degenerate. A zero in the precision matrix implies conditional independence of the corresponding nodes. Graphical lasso proposes to solve
\begin{align}
    f(\bOmega) + \rho \sum_{i < j} |\omega_{ij}|,   \label{eqn:glasso-obj}
\end{align}
where $\rho\geq0$ is the tuning constant and $\omega_{ij}$ denotes the $(i,j)$-element of $\bOmega$. It is well-known that the determinant function is log-concave \citep{MagnusNeudecker99MatrixBook}. Therefore the loss function $f$ (\ref{eqn:graph-obj}) is convex and the EPSODE algorithm applies to (\ref{eqn:glasso-obj}). \cite{Friedman08GLasso} proposed an efficient coordinate descent procedure for solving (\ref{eqn:glasso-obj}) at a fixed $\rho$. A recent attempt to approximate the whole solution path is made by \cite{Yuan08GraphLasso}. Again his path algorithm can be deemed as a primitive predictor-corrector method for approximating the ODE solution.

With symmetry in mind, we parameterize $\bOmega$ in terms of its lower triangular part by a $p(p+1)/2$ column vector $\bx$ and let $D\bOmega(\bx) = \frac{\partial \text{vec}\bOmega}{\partial (\text{vec}\bx)^t}$
be the corresponding $p^2$-by-$p(p+1)/2$ Jacobian matrix. Note $D \bOmega(\bx) \cdot \bx = \text{vec} \bOmega(\bx)$ and each row of $D\bOmega(\bx)$ has exactly one nonzero entry which equals unity. We list here the first three derivatives of $f$. The proof is straightforward using matrix calculus and omitted for brevity.
\begin{lemma}
\begin{enumerate}
\item The derivatives for the Gaussian graphical model (\ref{eqn:graph-obj}) with respect to $\bOmega$ are
\begin{align*}
    Df(\bOmega) &= df(\bOmega) = [\text{vec}(-\bOmega^{-1} + \bSigma)]^t \\
    D^2f(\bOmega) &= d^2f(\bOmega) = \bOmega^{-1} \otimes \bOmega^{-1}  \\
    D^3f(\bOmega) 
    &= - (\bI_n \otimes \bK_{nn} \otimes \bI_n) \\
    & \hspace{.2in} \cdot [\bOmega^{-1} \otimes \bOmega^{-1} \otimes \text{vec}(\bOmega^{-1}) + \text{vec}(\bOmega^{-1}) \otimes \bOmega^{-1} \otimes \bOmega^{-1}],
\end{align*}
where $\bK_{nn}$ is the commutation matrix \citep{MagnusNeudecker99MatrixBook}.
\item The derivatives for the Gaussian graphical model (\ref{eqn:graph-obj}) with respect to $\bx$ are
\begin{align*}
    Df(\bx) &= Df(\bOmega) \cdot D\bOmega(\bx)    \\
    \bH(\bx) = D^2f(\bx) &= [D\bOmega(\bx)]^t \cdot D^2f(\bOmega) \cdot D\bOmega(\bx)   \\
    D \bH(\bx) = D^3f(\bx) &= \{ [D\bOmega(\bx)]^t \otimes [D\bOmega(\bx)]^t\} \cdot D^3f(\bOmega) \cdot D\bOmega(\bx).
\end{align*}
\end{enumerate}
\end{lemma}
\noindent
When the covariance matrix $\hat \bSigma$ is nonsingular, EPSODE can be initiated either at $\rho=0$ or $\rho=\infty$. When $\hat \bSigma$ is singular, we start from $\rho=\infty$ and the extended version of EPSODE (\ref{eqn:sol-ode-alt}) should be used. If starting at $\rho=0$, the solution is initialized at $\hat \bSigma^{-1}$; If starting at $\rho=\infty$, the solution is initialized at $\text{diag}(\hat \sigma_{ii}^{-1})$.  Minimization of both the unpenalized and penalized objective function has to be performed over the convex cone of symmetric, positive semidefinite matrices, which is not explicitly incorporated in our path following algorithm. The next result ensures the positive definiteness of the path solution.
\begin{lemma}[Positive definiteness along the path]
The path solution $\bOmega(\rho)$ minimizes (\ref{eqn:glasso-obj}) over the convex cone of symmetric, positive semidefinite matrices.
\end{lemma}
\begin{proof}
Both symmetry and the stationarity condition (\ref{path_stationary}) are preserved by the path following. Therefore the path solution $\bOmega(\rho)$ constitutes a minimum of the penalized objective function (\ref{eqn:glasso-obj}). This implies that $\text{det}(\bOmega(\rho))>0$ otherwise $f(\bOmega(\rho))=\infty$, contradicting with the optimality of $\bOmega(\rho)$.
\end{proof}

We illustrate the path algorithm by the classical example of 88 students' scores on five math courses -- mechanics, vector, algebra, analysis, and statistics \cite[Table 1.2.1]{MardiaKentBibby79Book}. Figure \ref{fig:MKBscore_lassopah} displays the solution path from EPSODE. The top three edges chosen by lasso are analysis-algebra, statistics-algebra, and algebra-vector, matching the findings in \citep{Yuan08GraphLasso}.

\begin{figure}
$$
\begin{array}{c}
\includegraphics[width=3.5in]{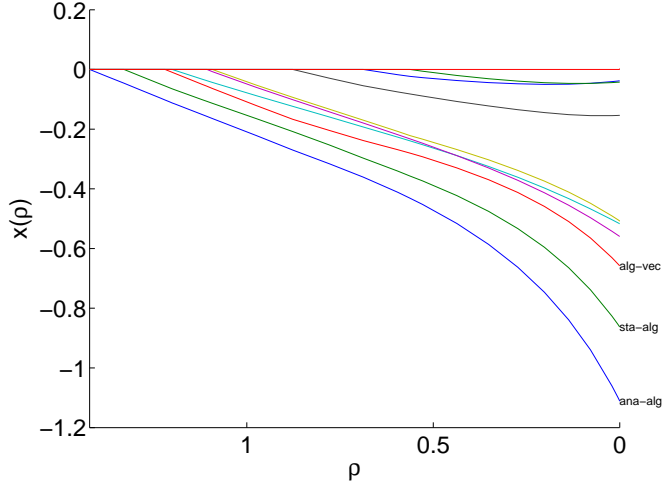}
\end{array}
$$
\caption{Solution path of the 10 edges in lasso-regularized Gaussian graphical model for the math score data. The top three edges chosen by lasso are labeled.}
\label{fig:MKBscore_lassopah}
\end{figure}

\subsection{Nonparametric Density Estimation}

As part of the trend of statistics shifting from parametric models to semi- or non-parametric models, nonparametric density estimation has attracted much attention in recent years. The maximum likelihood estimation for nonparametric density estimation often involves a nontrivial, high-dimensional constrained optimization problem. In this section, we briefly demonstrate the applicability of EPSODE to the maximum likelihood estimation of univariate log-concave density. Extensions to multivariate log-concave density estimation \citep{CuleGramacySamworth2009LogConcDEAD,CuleGramacySamworth2010LogConc} and $k$-monotone density estimation \citep{BalabdaouiWellner07kMono} will be pursued elsewhere. Some algorithms have been specifically crafted for log-concave density estimation, e.g., the iterative convex minorant algorithm (ICMA) \citep{GroeneboomWellner92Nonparam,Jongbloed98ICMA} and more recently an active-set algorithm \citep{DuembgenRufibachHuesler07LogConc}. It is noteworthy that, besides providing an alternative solver for log-concave density estimation, EPSODE offers the whole solution path between the unconstrained and constrained solutions. For example, an ``almost" log-concave density estimate in the middle of the path can be chosen that minimizes cross-validation or prediction error. This adds another dimension to the flexibility of nonparametric modeling.

The family of log-concave densities is an attractive modeling tool. It includes most of the commonly used parametric distributions as special cases. Examples include normal, gamma with shape parameter $\ge 1$, and beta densities with both parameters $\ge 1$. The survey paper \citep{Walther09Logconc} gives a recent review. A probability density $g(\cdot)$ on $\mathbb{R}$ is log-concave if its logarithm $\phi(x) = \ln g(x)$ is concave. Given iid observations, from an unknown distribution of density $g(\cdot)$, with support at points $x_1 < \ldots < x_n$ with corresponding frequencies $p_1, \ldots, p_n$, it is well-known \citep{Walther02Multiscale} that the nonparametric MLE of  $g$ exists, is unique and takes the form $\hat g = \exp (\hat \phi)$ where $\hat \phi$ is continuous and piecewise linear on $[x_1,x_n]$, with the set of knots contained in $\{x_1,\ldots,x_n\}$, and $\hat \phi = - \infty$ outside the interval $[x_1,x_n]$. This implies that the MLE is obtained by minimizing the strictly convex function
\begin{align*}
    f(\bphi) = - \sum_{i=1}^n p_i \phi_i + \sum_{k=1}^{n-1} (x_{k+1}-x_k) \int_0^1 e^{(1-t) \phi_k+t \phi_{k+1}} \, dt.
\end{align*}
over $\bphi=(\phi_1, \phi_2, \cdots, \phi_n)^t \in \mathbb{R}^n$ subject to constraints
\begin{align*}
    \frac{\phi_{i+1} - \phi_i}{x_{i+1} - x_i} \le \frac{\phi_{i} - \phi_{i-1}}{x_{i} - x_{i-1}}, \hspace{.1in} i=2,\ldots,n-1.
\end{align*}
The consistency of the MLE is proved by \cite{PalWoodroofeMeyer07Polya} and the pointwise asymptotic distribution of the MLE studied in \citep{BalabdaouiRufibachWellner09Logconc}.

Following \cite{DuembgenRufibachHuesler07LogConc}, we use notations
\begin{align*}
    \delta_0 &= \delta_{n} = 0, \, \delta_i = x_{i+1} - x_i, \hspace{.1in} i=1,\ldots,n-1   \\
    J(r,s) &= \int_0^1 e^{(1-t)r+ts} \, dt = \begin{cases}
    \frac{e^s - e^r}{s-r} & r \ne s \\
    e^r & r=s
    \end{cases}.
\end{align*}
Then the objective function becomes
\begin{align*}
    f(\bphi) = - \sum_{i=1}^n p_i \phi_i + \sum_{k=1}^{n-1} \delta_k J(\phi_k,\phi_{k+1}).
\end{align*}
The path algorithm requires up to the third derivative of the objective function $f$
\begin{align*}
    [\nabla f(\bphi)]_i &= - p_i + \delta_{i-1} J_{01}(\phi_{i-1},\phi_i) + \delta_{i} J_{10}(\phi_{i},\phi_{i+1}) \\
    [H(\bphi)]_{ij} &= [d^2f(\bphi)]_{ij}  \\
    &= \begin{cases}
    \delta_{i-1} J_{11}(\phi_{i-1},\phi_i) & j=i-1 \\
    \delta_{i-1} J_{02}(\phi_{i-1},\phi_i) + \delta_{i} J_{20}(\phi_{i},\phi_{i+1}) & j=i   \\
    \delta_{i} J_{11}(\phi_{i},\phi_{i+1}) & j=i+1    \\
    0 & \text{otherwise}
    \end{cases} \\
    \frac{\partial [H(\bphi)]_{i,i-1}}{\partial \phi_k} &= \begin{cases}
       \delta_{i-1} J_{21}(\phi_{i-1},\phi_i) & k=i-1 \\
       \delta_{i-1} J_{12}(\phi_{i-1},\phi_i) & k=i   \\
       0 & \text{otherwise}
    \end{cases} \\
    \frac{\partial [H(\bphi)]_{i,i}}{\partial \phi_k} &= \begin{cases}
       \delta_{i-1} J_{12}(\phi_{i-1},\phi_i) & k=i-1 \\
       \delta_{i-1} J_{03}(\phi_{i-1},\phi_i) + \delta_{i} J_{30}(\phi_{i},\phi_{i+1}) & k=i   \\
       \delta_{i} J_{21}(\phi_{i},\phi_{i+1}) & k=i+1   \\
       0 & \text{otherwise}
    \end{cases} \\
    \frac{\partial [H(\bphi)]_{i,i+1}}{\partial \phi_k} &= \begin{cases}
       \delta_{i} J_{21}(\phi_{i},\phi_{i+1}) & k=i \\
       \delta_{i} J_{12}(\phi_{i},\phi_{i+1}) & k=i+1   \\
       0 & \text{otherwise}
    \end{cases}.
\end{align*}
Interchanging the derivative and integral operators, justified by the dominated convergence theorem, gives a useful representation for the partial derivatives of $J$
\begin{align*}
    J_{ab}(r,s) &= \frac{\partial^{a+b}}{\partial r^a \partial s^b} J(r,s) = \int_0^1 (1-t)^a t^b e^{(1-t)r + ts} \, dt.
\end{align*}
We derive a recurrence relation for $J_{ab}(r,s)$ to facilitate its computation.
\begin{lemma}
\label{lemma:Jab}
$J_{ab}(r,s)$ satisfy following recurrence
\begin{enumerate}
\item For $r \ne s$,
\begin{align*}
    J_{00}(r,s) &= \frac{e^s - e^r}{s-r} \\
    J_{10}(r,s) &= - \frac{e^r}{s-r} + \frac{e^s-e^r}{(s-r)^2} \\
    J_{01}(r,s) &= \frac{e^s}{s-r} - \frac{e^s-e^r}{(s-r)^2} \\
    J_{11}(r,s) &= \frac{e^s+e^r}{(s-r)^2} - \frac{2(e^s-e^r)}{(s-r)^3}  \\
    J_{ab}(r,s) &= \frac{a+b+s-r}{s-r} J_{a-1,b}(r,s) - \frac{a-1}{s-r} J_{a-2,b}(r,s)  \\
    J_{ab}(r,s) &= - \frac{a+b-s+r}{s-r} J_{a,b-1}(r,s) + \frac{b-1}{s-r} J_{a,b-2}(r,s).
\end{align*}
\item For $r=s$,
\begin{align*}
    J_{ab}(r,s) = \frac{e^r a! b!}{(a+b+1)!} = \frac{a}{a+b+1} J_{a-1,b} = \frac{b}{a+b+1} J_{a,b-1}.
\end{align*}
\end{enumerate}
\end{lemma}
\begin{proof}
We recognize $J_{ab}(r,s)$ as the Kummer confluent hypergeometric function multiplied by a constant
\begin{align*}
    J_{ab}(r,s) &= e^rB(a+1,b+1) \sum_{k=0}^\infty \frac{(b+1)_{(k)}}{(a+b+2)_{(k)}} \frac{(s-r)^k}{k!}  \\
    &= e^rB(a+1,b+1) \, _1F_1(b+1,a+b+2 \mid s-r) \\
    &= e^sB(a+1,b+1) \, _1F_1(a+1,a+b+2 \mid r-s)  \\
    &= \frac{e^sa!b!}{(a+b+1)!} \, _1F_1(a+1,a+b+2 \mid r-s).
\end{align*}
Then the results follow from the well-known recurrence relation for Kummer hypergeometric function
\begin{align*}
    \,_1F_1(x,y\mid z) = \frac{(1-y)(y+z-2)}{(x-y+1)z} \,_1F_1(x,y-1\mid z) + \frac{(1-y)(2-y)}{(x-y+1)z} \,_1F_1(x,y-2 \mid z).
\end{align*}
and symmetry $J_{ab}(r,s) = J_{ba}(s,r)$.
\end{proof}
\noindent

To illustrate the path algorithm for this problem, we simulate $n=25$ points from the extremal distribution Gumbel(0,1). Figure \ref{fig:gumbel} displays the constrained and unconstrained estimates of $\phi_i$ and the solution path bridging the two.
\begin{figure}
$$
\begin{array}{cc}
\includegraphics[width=2.5in]{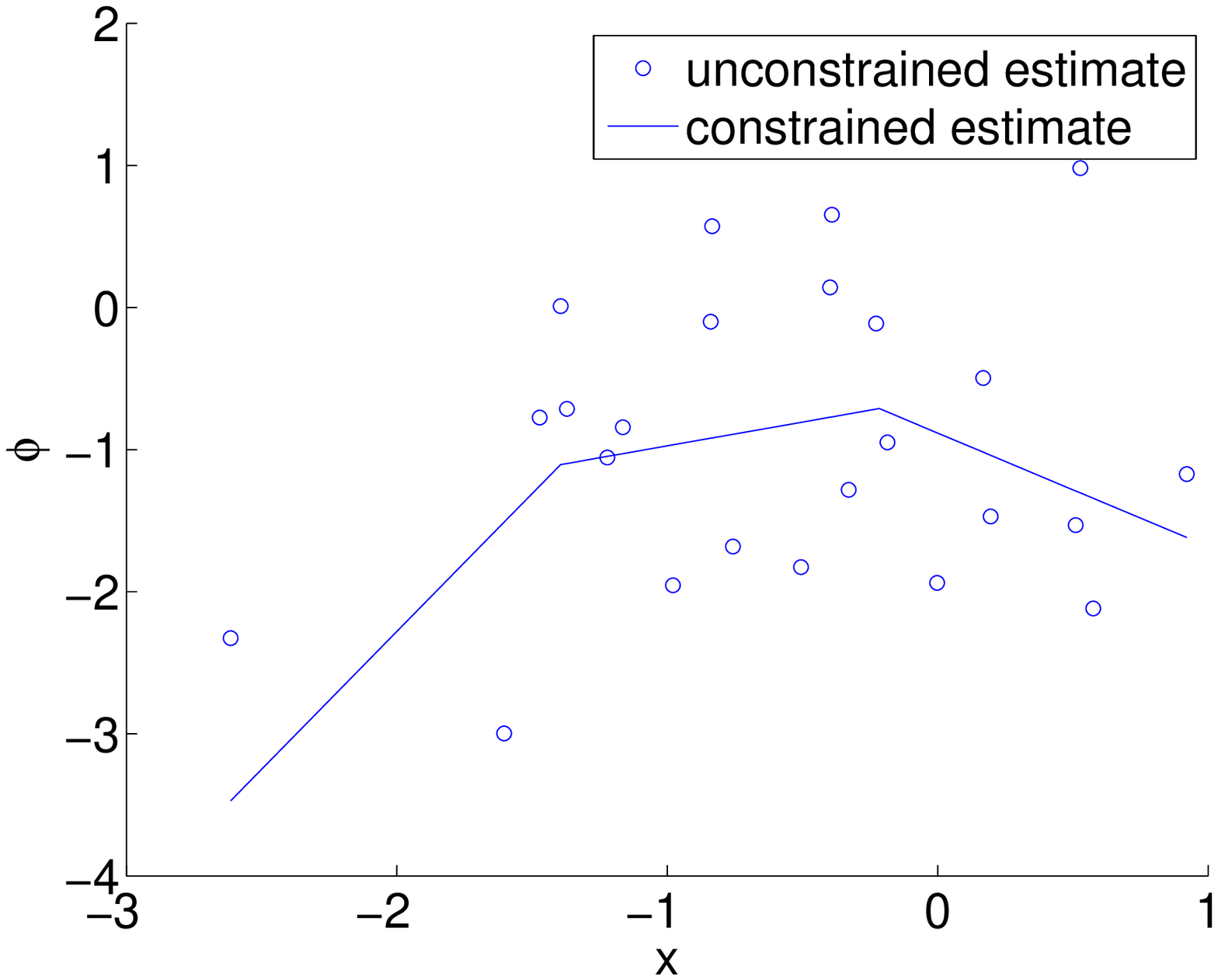} &
\includegraphics[width=2.5in]{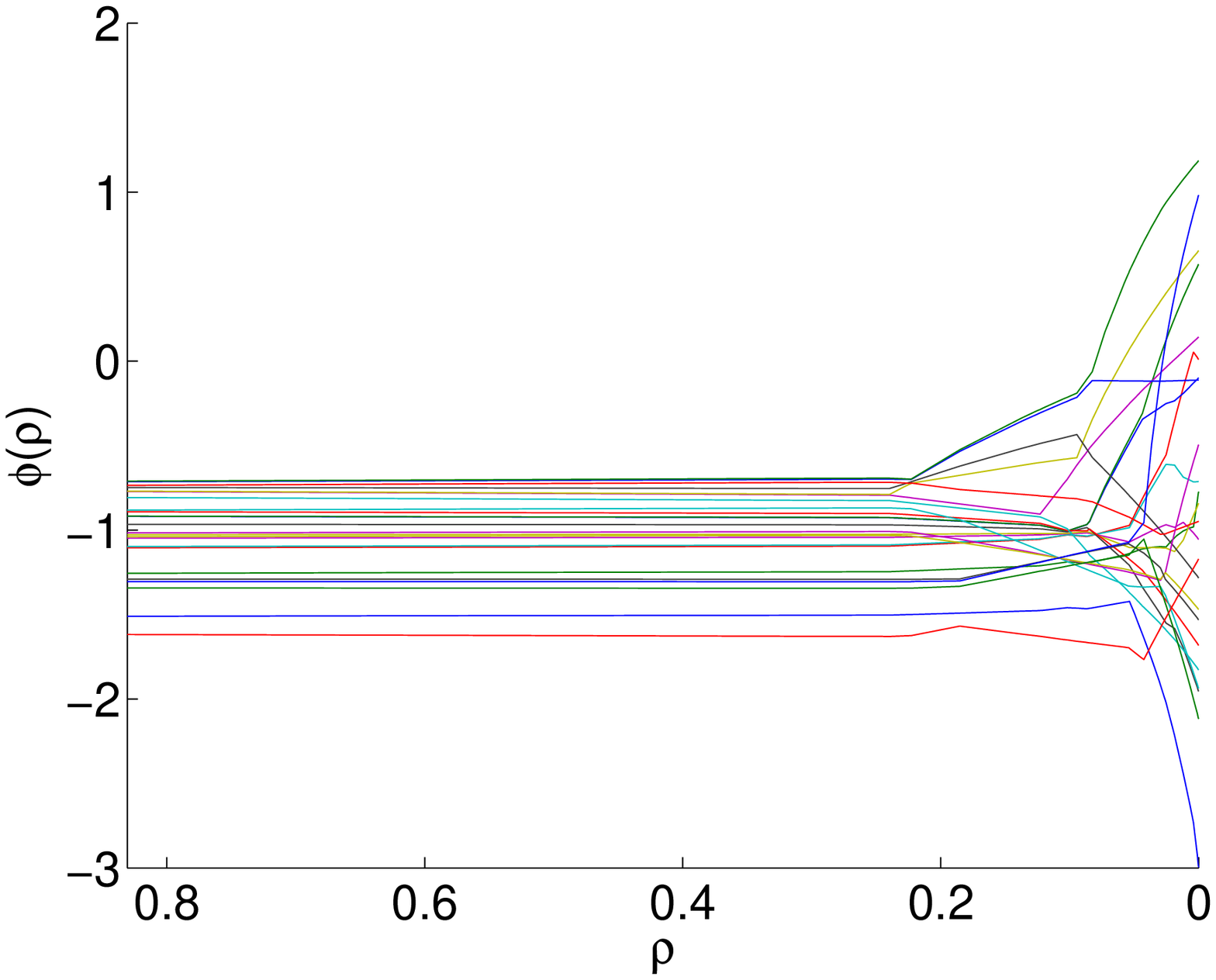}   \\
\includegraphics[width=2.5in]{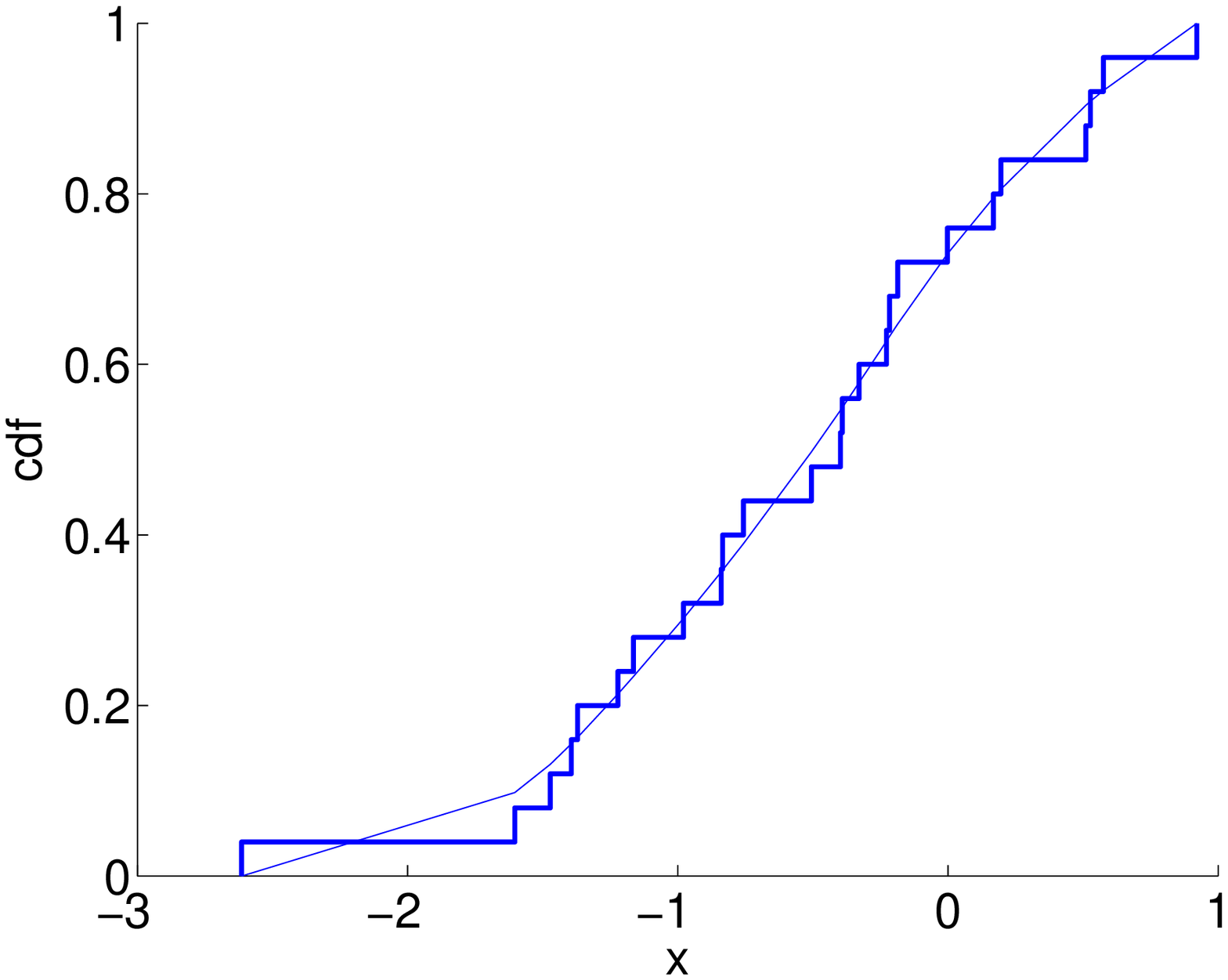}
\end{array}
$$
\caption{Log-concave density estimation. $n=25$ points are generated from Gumbel(0,1) distribution. Top left: Unconstrained and concavity-constrained estimates $\phi$. Top right: Solution path. Bottom left: Empirical cdf and the cdf of MLE density.}
\label{fig:gumbel}
\end{figure}

\section{Conclusions}
\label{sec:conclusions}

In this article we propose a generic path following algorithm EPSODE that works for any regularization problems of form (\ref{eqn:pathalgo-obj}). The advantages are its simplicity and generality. Path following only involves solving ODEs segment by segment and is simple to implement using popular softwares such as {\tt R} and {\tt Matlab}. Besides providing the whole regularization path, it also gives a solver for linearly constrained optimization problems that frequently arise in statistics. Our applications to shape-restricted regressions and nonparametric density estimation are special cases in particular.

At least two extensions deserve further study. Current algorithm requires sufficient smoothness (twice differentiable) in the loss function. This precludes certain applications with non-smooth objective function, e.g., the Huber loss in robust estimation and the loss function in quantile regression. Generalization of our path algorithm to regularization of these loss functions requires further research. Another restriction in our formulation is the linearity in the regularization terms. In sparse regressions, several authors have proposed nonlinear and  non-convex penalties. The bridge regression \citep{FrankFriedman93Bridge,Fu98BridgeLasso} and SCAD penalties \citep{FanLi01SCAD} fall into this category. As observed in \citep{Friedman08GPS}, when the penalty is not convex, the solution path may not be continuous and poses difficulty in path following, which strongly depends on the continuity and smoothness of the solution path. Fortunately, in these problems, the discontinuities only occur when new variables enter or leave the model. A promising strategy is to initialize the starting point of next segment by solving an equality constrained optimization problem. This again invites further investigation.

\bibliography{../../../bib-HZ}

\begin{thebibliography}{}

\bibitem[\protect\citeauthoryear{Agresti}{Agresti}{2002}]{Agresti02Book}
Agresti, A. (2002).
\newblock {\em Categorical Data Analysis\/} (Second ed.).
\newblock Wiley Series in Probability and Statistics. New York:
  Wiley-Interscience [John Wiley \& Sons].

\bibitem[\protect\citeauthoryear{Balabdaoui, Rufibach, and Wellner}{Balabdaoui
  et~al.}{2009}]{BalabdaouiRufibachWellner09Logconc}
Balabdaoui, F., K.~Rufibach, and J.~A. Wellner (2009).
\newblock Limit distribution theory for maximum likelihood estimation of a
  log-concave density.
\newblock {\em Ann. Statist.\/}~{\em 37\/}(3), 1299--1331.

\bibitem[\protect\citeauthoryear{Balabdaoui and Wellner}{Balabdaoui and
  Wellner}{2007}]{BalabdaouiWellner07kMono}
Balabdaoui, F. and J.~A. Wellner (2007).
\newblock Estimation of a {$k$}-monotone density: limit distribution theory and
  the spline connection.
\newblock {\em Ann. Statist.\/}~{\em 35\/}(6), 2536--2564.

\bibitem[\protect\citeauthoryear{Chen, Donoho, and Saunders}{Chen
  et~al.}{2001}]{ChenDonohoSaunders01BasisPursuit}
Chen, S.~S., D.~L. Donoho, and M.~A. Saunders (2001).
\newblock Atomic decomposition by basis pursuit.
\newblock {\em SIAM Rev.\/}~{\em 43\/}(1), 129--159.

\bibitem[\protect\citeauthoryear{Chen, Lin, Kim, Carbonell, and Xing}{Chen
  et~al.}{2010}]{ChenLinKimCarbonellXing10ProxGrad}
Chen, X., Q.~Lin, S.~Kim, J.~Carbonell, and E.~Xing (2010).
\newblock An efficient proximal gradient method for general structured sparse
  learning.
\newblock {\em arXiv:1005.4717\/}.

\bibitem[\protect\citeauthoryear{Cule, Gramacy, and Samworth}{Cule
  et~al.}{2009}]{CuleGramacySamworth2009LogConcDEAD}
Cule, M., R.~B. Gramacy, and R.~Samworth (2009, 1).
\newblock {LogConcDEAD}: An {R} package for maximum likelihood estimation of a
  multivariate log-concave density.
\newblock {\em Journal of Statistical Software\/}~{\em 29\/}(2), 1--20.

\bibitem[\protect\citeauthoryear{Cule, Samworth, and Stewart}{Cule
  et~al.}{2010}]{CuleGramacySamworth2010LogConc}
Cule, M., R.~Samworth, and M.~Stewart (2010).
\newblock Maximum likelihood estimation of a multi-dimensional log-concave
  density.
\newblock {\em Journal of the Royal Statistical Society Series B\/}~{\em
  72\/}(5), 545--607.

\bibitem[\protect\citeauthoryear{Dempster}{Dempster}{1969}]{Dempster69Book}
Dempster, A.~P. (1969).
\newblock {\em Elements of Continuous Multivariate Analysis}.
\newblock Addison-Wesley series in behavioral sciences. Reading, MA:
  Addison-Wesley.

\bibitem[\protect\citeauthoryear{Duembgen, Rufibach, and Huesler}{Duembgen
  et~al.}{2007}]{DuembgenRufibachHuesler07LogConc}
Duembgen, L., K.~Rufibach, and A.~Huesler (2007).
\newblock Active set and {EM} algorithms for log-concave densities based on
  complete and censored data.

\bibitem[\protect\citeauthoryear{Efron, Hastie, Johnstone, and
  Tibshirani}{Efron et~al.}{2004}]{EfronHastieIainTibshirani04LARS}
Efron, B., T.~Hastie, I.~Johnstone, and R.~Tibshirani (2004).
\newblock Least angle regression.
\newblock {\em Ann. Statist.\/}~{\em 32\/}(2), 407--499.
\newblock With discussion, and a rejoinder by the authors.

\bibitem[\protect\citeauthoryear{Fan and Li}{Fan and Li}{2001}]{FanLi01SCAD}
Fan, J. and R.~Li (2001).
\newblock Variable selection via nonconcave penalized likelihood and its oracle
  properties.
\newblock {\em J. Amer. Statist. Assoc.\/}~{\em 96\/}(456), 1348--1360.

\bibitem[\protect\citeauthoryear{Fan, Maity, Wang, and Wu}{Fan
  et~al.}{2011}]{FanMaityWangWu11ManA}
Fan, J., A.~Maity, Y.~Wang, and Y.~Wu (2011).
\newblock Analyzing mergers and acquisition data using parametrically guided
  generalized additive models.
\newblock ~{\em submitted}.

\bibitem[\protect\citeauthoryear{Frank and Friedman}{Frank and
  Friedman}{1993}]{FrankFriedman93Bridge}
Frank, I.~E. and J.~H. Friedman (1993).
\newblock A statistical view of some chemometrics regression tools.
\newblock {\em Technometrics\/}~{\em 35\/}(2), 109--135.

\bibitem[\protect\citeauthoryear{Friedman}{Friedman}{2008}]{Friedman08GPS}
Friedman, J. (2008).
\newblock Fast sparse regression and classification.
\newblock {\em http://www-stat.stanford.edu/~jhf/ftp/GPSpaper.pdf\/}.

\bibitem[\protect\citeauthoryear{Friedman, Hastie, and Tibshirani}{Friedman
  et~al.}{2000}]{friedman00additive}
Friedman, J., T.~Hastie, and R.~Tibshirani (2000).
\newblock Additive logistic regression: a statistical view of boosting.
\newblock {\em The Annals of Statistics\/}~{\em 28\/}(2), 337--407.

\bibitem[\protect\citeauthoryear{Friedman, Hastie, and Tibshirani}{Friedman
  et~al.}{2008}]{Friedman08GLasso}
Friedman, J., T.~Hastie, and R.~Tibshirani (2008, July).
\newblock {Sparse inverse covariance estimation with the graphical lasso}.
\newblock {\em Biostatistics\/}~{\em 9\/}(3), 432--441.

\bibitem[\protect\citeauthoryear{Fu}{Fu}{1998}]{Fu98BridgeLasso}
Fu, W.~J. (1998).
\newblock Penalized regressions: the bridge versus the lasso.
\newblock {\em Journal of Computational and Graphical Statistics\/}~{\em
  7\/}(3), 397--416.

\bibitem[\protect\citeauthoryear{Goodnight}{Goodnight}{1979}]{Goodnight79Sweep}
Goodnight, J.~H. (1979).
\newblock A tutorial on the sweep operator.
\newblock {\em Amer. Statist.\/}~{\em 33\/}(3), pp. 149--158.

\bibitem[\protect\citeauthoryear{Groeneboom and Wellner}{Groeneboom and
  Wellner}{1992}]{GroeneboomWellner92Nonparam}
Groeneboom, P. and J.~A. Wellner (1992).
\newblock {\em Information bounds and nonparametric maximum likelihood
  estimation}, Volume~19 of {\em DMV Seminar}.
\newblock Basel: Birkh\"auser Verlag.

\bibitem[\protect\citeauthoryear{Hastie and Tibshirani}{Hastie and
  Tibshirani}{1993}]{HastieTibshirani93VaryingCoeff}
Hastie, T. and R.~Tibshirani (1993).
\newblock Varying-coefficient models.
\newblock {\em J. Roy. Statist. Soc. Ser. B\/}~{\em 55\/}(4), 757--796.
\newblock With discussion and a reply by the authors.

\bibitem[\protect\citeauthoryear{Jennrich}{Jennrich}{1977}]{Jennrich77Stepwise%
reg}
Jennrich, R. (1977).
\newblock Stepwise regression.
\newblock In {\em Statistical Methods for Digital Computers}, pp.\  58--75. New
  York: Wiley-Interscience.

\bibitem[\protect\citeauthoryear{Jongbloed}{Jongbloed}{1998}]{Jongbloed98ICMA}
Jongbloed, G. (1998).
\newblock The iterative convex minorant algorithm for nonparametric estimation.
\newblock {\em J. Comput. Graph. Statist.\/}~{\em 7\/}(3), 310--321.

\bibitem[\protect\citeauthoryear{Kim, Koh, Boyd, and Gorinevsky}{Kim
  et~al.}{2009}]{KimKohBoyd09TrendFiltering}
Kim, S.-J., K.~Koh, S.~Boyd, and D.~Gorinevsky (2009).
\newblock {$l_1$} trend filtering.
\newblock {\em SIAM Rev.\/}~{\em 51\/}(2), 339--360.

\bibitem[\protect\citeauthoryear{Lange}{Lange}{2010}]{Lange10NumAnalBook}
Lange, K. (2010).
\newblock {\em Numerical Analysis for Statisticians\/} (Second ed.).
\newblock Statistics and Computing. New York: Springer.

\bibitem[\protect\citeauthoryear{Lawson and Hanson}{Lawson and
  Hanson}{1987}]{LawsonHanson87LSBook}
Lawson, C.~L. and R.~J. Hanson (1987).
\newblock {\em Solving Least Squares Problems\/} (New edition ed.).
\newblock Classics in Applied Mathematics. Society for Industrial Mathematics.

\bibitem[\protect\citeauthoryear{Little and Rubin}{Little and
  Rubin}{2002}]{LittleRubin02Book}
Little, R. J.~A. and D.~B. Rubin (2002).
\newblock {\em Statistical Analysis with Missing Data\/} (Second ed.).
\newblock Wiley Series in Probability and Statistics. Hoboken, NJ:
  Wiley-Interscience [John Wiley \& Sons].

\bibitem[\protect\citeauthoryear{Magnus and Neudecker}{Magnus and
  Neudecker}{1999}]{MagnusNeudecker99MatrixBook}
Magnus, J.~R. and H.~Neudecker (1999).
\newblock {\em Matrix Differential Calculus with Applications in Statistics and
  Econometrics}.
\newblock Wiley Series in Probability and Statistics. Chichester: John Wiley \&
  Sons Ltd.

\bibitem[\protect\citeauthoryear{Mardia, Kent, and Bibby}{Mardia
  et~al.}{1979}]{MardiaKentBibby79Book}
Mardia, K.~V., J.~T. Kent, and J.~M. Bibby (1979).
\newblock {\em Multivariate Analysis}.
\newblock London: Academic Press [Harcourt Brace Jovanovich Publishers].
\newblock Probability and Mathematical Statistics: A Series of Monographs and
  Textbooks.

\bibitem[\protect\citeauthoryear{McCullagh and Nelder}{McCullagh and
  Nelder}{1983}]{McCullaghNelder83GLMBook}
McCullagh, P. and J.~A. Nelder (1983).
\newblock {\em Generalized Linear Models}.
\newblock Monographs on Statistics and Applied Probability. London: Chapman \&
  Hall.

\bibitem[\protect\citeauthoryear{Nocedal and Wright}{Nocedal and
  Wright}{2006}]{NocedalWright06Book}
Nocedal, J. and S.~J. Wright (2006).
\newblock {\em Numerical Optimization\/} (Second ed.).
\newblock Springer Series in Operations Research and Financial Engineering. New
  York: Springer.

\bibitem[\protect\citeauthoryear{Osborne, Presnell, and Turlach}{Osborne
  et~al.}{2000}]{OsbornePresnellTurlach00LassoAlgo}
Osborne, M.~R., B.~Presnell, and B.~A. Turlach (2000).
\newblock A new approach to variable selection in least squares problems.
\newblock {\em IMA J. Numer. Anal.\/}~{\em 20\/}(3), 389--403.

\bibitem[\protect\citeauthoryear{Pal, Woodroofe, and Meyer}{Pal
  et~al.}{2007}]{PalWoodroofeMeyer07Polya}
Pal, J.~K., M.~Woodroofe, and M.~Meyer (2007).
\newblock Estimating a {P}olya frequency function{${}_2$}.
\newblock In {\em Complex datasets and inverse problems}, Volume~54 of {\em IMS
  Lecture Notes Monogr. Ser.}, pp.\  239--249. Beachwood, OH: Inst. Math.
  Statist.

\bibitem[\protect\citeauthoryear{Park and Hastie}{Park and
  Hastie}{2007}]{ParkHastie07GLMLasso}
Park, M.~Y. and T.~Hastie (2007).
\newblock {$L_1$}-regularization path algorithm for generalized linear models.
\newblock {\em J. R. Stat. Soc. Ser. B Stat. Methodol.\/}~{\em 69\/}(4),
  659--677.

\bibitem[\protect\citeauthoryear{Robertson, Wright, and Dykstra}{Robertson
  et~al.}{1988}]{RobertsonWrightDykstra88Book}
Robertson, T., F.~T. Wright, and R.~L. Dykstra (1988).
\newblock {\em Order Restricted Statistical Inference}.
\newblock Wiley Series in Probability and Mathematical Statistics: Probability
  and Mathematical Statistics. Chichester: John Wiley \& Sons Ltd.

\bibitem[\protect\citeauthoryear{Rosset and Zhu}{Rosset and
  Zhu}{2007}]{RossetZhu07Path}
Rosset, S. and J.~Zhu (2007).
\newblock Piecewise linear regularized solution paths.
\newblock {\em Ann. Statist.\/}~{\em 35\/}(3), 1012--1030.

\bibitem[\protect\citeauthoryear{Rufibach}{Rufibach}{2010}]{Rufibach10GLMOrder}
Rufibach, K. (2010, June).
\newblock An active set algorithm to estimate parameters in generalized linear
  models with ordered predictors.
\newblock {\em Computational Statistics \& Data Analysis\/}~{\em 54\/}(6),
  1442--1456.

\bibitem[\protect\citeauthoryear{Ruszczy{\'n}ski}{Ruszczy{\'n}ski}{2006}]{Rusz%
czynski06Book}
Ruszczy{\'n}ski, A. (2006).
\newblock {\em Nonlinear Optimization}.
\newblock Princeton, NJ: Princeton University Press.

\bibitem[\protect\citeauthoryear{Shivdasani and Wang}{Shivdasani and
  Wang}{2009}]{ShivdasaniWang09}
Shivdasani, A. and Y.~Wang (2009).
\newblock Did structured credit fuel and {LBO} boom?
\newblock ~{\em http://ssrn.com/abstract=1394421}.

\bibitem[\protect\citeauthoryear{Silvapulle and Sen}{Silvapulle and
  Sen}{2005}]{SilvapullePranab05CSIBook}
Silvapulle, M.~J. and P.~K. Sen (2005).
\newblock {\em Constrained Statistical Inference: Inequality, Order, and Shape
  Restrictions}.
\newblock Wiley Series in Probability and Statistics. Hoboken, NJ:
  Wiley-Interscience [John Wiley \& Sons].

\bibitem[\protect\citeauthoryear{Tibshirani}{Tibshirani}{1996}]{Tibshirani96La%
sso}
Tibshirani, R. (1996).
\newblock Regression shrinkage and selection via the lasso.
\newblock {\em J. Roy. Statist. Soc. Ser. B\/}~{\em 58\/}(1), 267--288.

\bibitem[\protect\citeauthoryear{Tibshirani, Saunders, Rosset, Zhu, and
  Knight}{Tibshirani et~al.}{2005}]{Tibshirani05FusedLasso}
Tibshirani, R., M.~Saunders, S.~Rosset, J.~Zhu, and K.~Knight (2005).
\newblock Sparsity and smoothness via the fused lasso.
\newblock {\em J. R. Stat. Soc. Ser. B Stat. Methodol.\/}~{\em 67\/}(1),
  91--108.

\bibitem[\protect\citeauthoryear{Tibshirani and Taylor}{Tibshirani and
  Taylor}{2011}]{TibshiraniTaylor10GenLasso}
Tibshirani, R. and J.~Taylor (2011).
\newblock The solution path of the generalized lasso.
\newblock {\em Ann. Statist.\/}~{\em to appear}.

\bibitem[\protect\citeauthoryear{Walther}{Walther}{2002}]{Walther02Multiscale}
Walther, G. (2002).
\newblock Detecting the presence of mixing with multiscale maximum likelihood.
\newblock {\em J. Amer. Statist. Assoc.\/}~{\em 97\/}(458), 508--513.

\bibitem[\protect\citeauthoryear{Walther}{Walther}{2009}]{Walther09Logconc}
Walther, G. (2009).
\newblock Inference and modeling with log-concave distributions.
\newblock {\em Statist. Sci.\/}~{\em 24\/}(3), 319--327.

\bibitem[\protect\citeauthoryear{Wu}{Wu}{2011}]{Wu10ODELasso}
Wu, Y. (2011).
\newblock An ordinary differential equation-based solution path algorithm.
\newblock {\em Journal of Nonparametric Statistics\/}~{\em 23}, 185--199.

\bibitem[\protect\citeauthoryear{Yuan}{Yuan}{2008}]{Yuan08GraphLasso}
Yuan, M. (2008).
\newblock Efficient computation of {$\ell_1$} regularized estimates in
  {G}aussian graphical models.
\newblock {\em J. Comput. Graph. Statist.\/}~{\em 17\/}(4), 809--826.

\bibitem[\protect\citeauthoryear{Zhou and Lange}{Zhou and
  Lange}{2011}]{ZhouLange11LSPath}
Zhou, H. and K.~Lange (2011).
\newblock A path algorithm for constrained estimation.
\newblock {\em arXiv:1103.3738\/}.

\bibitem[\protect\citeauthoryear{Zou, Hastie, and Tibshirani}{Zou
  et~al.}{2007}]{ZouHastieTibshirani07LassoDF}
Zou, H., T.~Hastie, and R.~Tibshirani (2007).
\newblock On the ``degrees of freedom'' of the lasso.
\newblock {\em Ann. Statist.\/}~{\em 35\/}(5), 2173--2192.

\end{thebibliography}
\bibliographystyle{Chicago}

\end{document}